%% file: newms.tex
\pdfoutput=1

\RequirePackage[l2tabu,orthodox]{nag}

\documentclass[11pt,letterpaper,reqno,oneside]{article}

\input{preamble.tex}

\title{\scshape Misperception and informativeness\\in statistical discrimination\thanks{We are grateful for useful comments from Arjada Bardhi, Aislinn Bohren, Davide Bordoli, Henrique Castro-Pires, Chris Chambers, Gregorio Curello, Federico Echenique, Ben Golub, Ian Jewitt, Meg Meyer, Marco Ottaviani, Augustus Smith, and the audience at the LSE Perspectives on Economic Theory conference. Oscar Calvert provided excellent research assistance.}}

\author{%
\begin{tabular}{cc}
	Matteo Escudé
	& Paula Onuchic \\
	LUISS
	& London School of Economics \vspace{1em} \\
	\parbox{\widthof{Quitzé Valenzuela-Stookey}}{\centering Ludvig Sinander}
	& Quitzé Valenzuela-Stookey \\
	University of Oxford
	& UC Berkeley
\end{tabular}%
}

\date{18 June 2026}

\makeatletter
	\AtBeginDocument{ \hypersetup{ 
		pdftitle = {Misperception and informativeness in statistical discrimination}, 
		pdfauthor = {Matteo Escudé, Paula Onuchic, Ludvig Sinander and Quitzé Valenzuela-Stookey} 
		} }
\makeatother

\begin{document}

\maketitle

\begin{abstract}
We study the interplay of information and prior (mis)perceptions in a Phelps--Aigner--Cain-type model of statistical discrimination in the labor market. We decompose the effect on average pay of an increase in how informative observables are about workers' skills into a non-negative \emph{instrumental} component, reflecting increased surplus due to better matching of workers with tasks, and a \emph{perception-correcting} component capturing how extra information diminishes the importance of prior misperceptions about the distribution of skills in the worker population. We sign the perception-correcting term: it is non-negative (non-positive) if the population was ex-ante under-perceived (over-perceived). We then consider the implications for pay gaps between equally-skilled populations that differ in information, perceptions, or both, and identify conditions under which improving information narrows pay gaps.
\end{abstract}

\section{Introduction}
\label{sec:intro}

There are significant pay gaps between populations, for example between women and men and between different ethnic groups.%
\footnote{See e.g. \textcite{DoL2025} for the United States and \textcite{ONS2023ethnicity,ONS2024disability,ONS2024gender} for the United Kingdom.}
Possible contributing factors include discrimination, differences in human capital (potentially caused by discrimination earlier in life), and differences in individual behavior, such as occupation choice (also potentially influenced by discrimination).

The literature on \emph{statistical discrimination} argues that these pay gaps are partly explained by information-economic forces.%
\footnote{See e.g. \textcite{LangLehmann2012,AzmatPetrongolo2016,BertrandDuflo2017,BlauKahn2017,Neumark2018,LangSpitzer2020}. A related but distinct literature on ``taste-based'' discrimination \parencite{Becker1955,Becker1957} seeks explanations rooted in decision-makers' biased preferences, rather than in information.}
Statistical discrimination is possible whenever workers' skill is imperfectly observable. One way in which this can produce pay gaps is through \emph{misperception,} whereby firms believe that one population is less skilled on average than is actually the case. For example, users of the website StackExchange appear to harbor misperceptions about women's and men's abilities \parencite{BohrenImasRosenberg2019}.%
\footnote{Similarly, \textcite{AganStarr2018,ArnoldDobbieYang2018} present evidence of misperceptions (held by employers and by bail judges, respectively) about racial differences in the rates of commission of and conviction for crime in the United States.}
Another way in which skill unobservability can generate pay gaps is through differential \emph{informativeness,} when workers' observable characteristics (e.g. CVs and test scores) are less predictive of skill in one population than in another. For example, SAT scores are less predictive of academic ability for poorer students than for richer ones \parencite{Rothstein2004}.

In this paper, we study the interplay between misperception and informativeness in determining pay gaps, using a simple but general labor-market model in the spirit of the statistical-discrimination literature following \textcite{Phelps1972book,Phelps1972,AignerCain1977}.%
\footnote{Specifically, we adapt the model of \textcite{ChambersEchenique2021}.}
The most-emphasized result in the literature \parencite[see][]{AignerCain1977} is that a population's average pay is higher if its observable characteristics are more informative about skill. This finding relies on the assumption that perceptions are accurate. As for varying the perceptions themselves, it is intuitive and in fact true that a population is paid more on average the more favorably it is perceived.

Our starting point is the observation that these two effects interact: in particular, a population that is \emph{both} more informative \emph{and} more favorably perceived may be paid \emph{strictly less} on average. To see why, consider an economy in which every worker is either a ``high type'' with productivity~$1$ or a ``low type'' with productivity~$0$. Let $p \in (0,1)$ and $q \in (0,1)$ denote, respectively, the true and perceived fractions of high types. Assume that labor demand is competitive, so that each worker's pay equals her posterior expected productivity. If observable characteristics are completely uninformative about productivity, then all workers are paid $q$. If observables perfectly reveal productivity, then high types are paid $1$ and low types are paid $0$, so average pay is $p \cdot 1 + (1-p) \cdot 0 = p$. Thus if the population is over-perceived, i.e. $q>p$, then average pay is strictly \emph{lower} if observables perfectly reveal skill than if they are completely uninformative. The reason why this happens is that information corrects misperceptions, which is harmful for an over-perceived population. The implication for pay gaps is that between a population~$I$ with perception $q_I$ and completely informative observable characteristics and a population~$J$ with perception $q_J$ and completely uninformative characteristics, the pay gap $p-q_J$ is negative whenever population~$J$ is over-perceived ($q_J>p$), even if population~$I$ is more favorably perceived than population~$J$ ($q_I > q_J$). In other words, misperceptions can reverse the usual positive relationship between a population's average pay and the informativeness of its observable characteristics.

Our main theoretical contribution is to unpack this interaction between information and misperception. In particular, we decompose the change in a population's average pay when observables become more informative about skill into two terms. The first term, which we call \emph{instrumental,} captures how extra information is used to improve the assignment of workers to tasks, thereby increasing expected surplus and (thus) pay.%
\footnote{This term is zero in the example in the previous paragraph, since that example involves no task assignment.}
The second term, the \emph{perception-correcting component,} captures how the increased availability of information about workers diminishes the influence of prior misperceptions about the distribution of skills in the population. This term can be either negative or positive. (From the previous paragraph, we know that it can be negative enough to dominate the never-negative instrumental term.) Our theorem signs the perception-correcting component, showing that it is non-negative if the population was ex-ante under-perceived, non-positive if it was over-perceived, and zero if it was accurately perceived.

We then apply our decomposition to statistical discrimination, meaning pay gaps between populations whose (true) skill distributions are equal. In particular, we identify conditions under which a population~$I$ that is both more informative than and more favorably perceived than another population~$J$ will be paid more on average. The key condition turns out to be whether population~$J$ is under-perceived (relative to the true skill distribution): if yes, then population~$I$ will indeed earn more on average, and if no, then the reverse may occur (as shown above by example).

The spirit of the perception-correcting term is that prior perceptions about a population matter less when more information is available about individual workers. This suggests that the pay gap between two populations will narrow when more information is made available about workers in both populations, for example due to a policy intervention or technological progress. \textcite{BohrenImasRosenberg2019} confirm this intuition in a model of subjective evaluation (rather than of pay) that can be viewed as a special case of our model.%
\footnote{This result is not the main point of \textcite{BohrenImasRosenberg2019}: the authors have two other theoretical results, and their focus is on testing these empirically.}
We assess the validity of the ``information narrows pay gaps'' intuition beyond this special case, identifying (stringent) conditions under which it is correct, and showing that extra information may \emph{widen} pay gaps when these conditions do not hold.

This result has implications for the effectiveness of policy interventions that aim to reduce labor-market inequality by increasing information availability, such as testing regulations, transparency requirements, and the introduction of algorithmic screening tools. Such interventions are often motivated by the view that reducing informational asymmetries will both improve efficiency and narrow pay gaps. Our analysis shows that this intuition is incomplete: in the presence of misperceptions, adding information may worsen inequality even as it improves allocative efficiency. The actual welfare implications of informativeness-increasing policies depend on the nature of any misperceptions and on how new information reshapes task assignment.

\subsection{Related literature}
\label{sec:intro:lit}

We contribute to the statistical-discrimination literature initiated by \textcite{Phelps1972book,Phelps1972,AignerCain1977}, which shows how pay gaps can arise from firms' inference problem of estimating workers' unobservable skills from observables such as test scores and CVs.%
\footnote{One active strand of this literature concerns tests for distinguishing empirically between statistical and taste-based discrimination \parencite[e.g.][]{MartinMarx2021,Marx2022,BharadwajDebRenou2024,GaeblerGoel2025}. A related but distinct literature following \textcite{Arrow1973}, also called ``statistical discrimination,'' studies pay gaps arising from ``bad equilibria'' rather than from inference. For an overview of both theoretical literatures, see \textcite{FangMoro2011,Onuchic2025}.}
Within this literature, we take inspiration from \textcite{ChambersEchenique2021}, borrowing (and slightly enriching) their simple but general ``Phelpsian'' model of the labor market.

Most work on statistical discrimination has focused on the impact on pay of differential informativeness of observables about skill, under the assumption that firms accurately perceive the distribution of skills.%
\footnote{See the systematic literature review by \textcite{BohrenHaggagImasPope2025}.}
More recent work has argued for the importance of misperceptions in explaining empirical patterns of discrimination \parencite[e.g.][]{BohrenImasRosenberg2019,BohrenHaggagImasPope2025}.%
\footnote{One could also consider misperceptions about the structure of correlation between observables and skill. We discuss this possibility in \cref{sec:concl}.}
Our analysis elucidates how informativeness and misperception interact to determine pay gaps.

The instrumental component in our decomposition captures the familiar instrumental value of information, defined and characterized by \textcite{Blackwell1951,Blackwell1953}.%
\footnote{Blackwell's theorem is stated in \cref{sec:disc:accurate} below.}
The idea of a perception-correcting effect appears in the statistics literature,%
\footnote{See \textcite{Laplace1774,Wald1947,Girsanov1960,BlackwellDubins1962}.}
and plays a role in, for example, persuasion \parencite[e.g.][]{AlonsoCamara2016,OnuchicRay2023} and behavioral decision theory \parencite{Bordoli2025}. Bordoli's paper in particular distinguishes (informally) between instrumental and non-instrumental effects of new information.%
\footnote{See also \textcite{Morris1991,MorrisShin1997,Braghieri2023,Whitmeyer2024}.}
Our result on signing the perception-correcting component is closely related to \textcite{KartikLeeSuen2021}; in particular, it may be viewed as generalizing these authors' theorem to allow for task assignment.

\textcite{LiangLuMuOkumura2025} study statistical discrimination in a model much like ours, but assuming \emph{accurate} perceptions. They consider how average pay varies with task assignment and informativeness,%
\footnote{They actually interpret their model and results as describing algorithmic prediction (e.g. bail decisions), rather than task assignment and pay.}
identifying a trade-off between efficiency (higher average pay for both populations) and fairness (smaller pay gap) and characterizing the resulting frontier.

\subsection{Roadmap}
\label{sec:intro:roadmap}

\Cref{sec:model} describes the model, which predicts a population's average pay at a firm as a function of its (true and) perceived skill distribution and the informativeness of its observable characteristics about skill, as illustrated in \Cref{table:sum}.
\begin{figure}
\centering
\begin{tabular}{rrcc}
 & & \multicolumn{2}{c}{$\overset{\raisebox{3pt}{\text{\footnotesize information}}}{\rotatebox[origin=c]{180}{$\underbrace{\hspace{\widthof{\footnotesize information}}}$}}$} \vspace{4pt} \\
 & \multicolumn{1}{r|}{} & low & high \\ \cline{2-4}
\multirow{2}{\widthof{\footnotesize perception $\bigg\{$}}{\footnotesize perception $\bigg\{$} & \multicolumn{1}{r|}{unfavorable} & $w_1$ & $w_3$ \\
 & \multicolumn{1}{r|}{favorable} & $w_2$ & $w_4$
\end{tabular}
\caption{Illustration of the model and results.}
\label{table:sum}
\end{figure}
In \cref{sec:model:lemma}, we present a preliminary lemma asserting that more favorable perception increases average pay, i.e. $w_1 \leq w_2$ (and $w_3 \leq w_4$) in \Cref{table:sum}. Our central theorem, in \cref{sec:decomp}, decomposes the effect $w_3-w_1$ (or $w_4-w_2$) of new information, and signs the terms. In \cref{sec:disc}, we combine the previous results to sign $w_4-w_1$, interpreted as the pay gap between two equally-skilled populations with different perceptions and informativeness. Finally, in \cref{sec:gap}, we investigate whether $w_2-w_1 \geq w_4-w_3$, i.e. whether extra information narrows the pay gap between two equally-skilled but differently perceived populations. All proofs are in the appendix.

\section{Model}
\label{sec:model}

We consider a model in the spirit of \textcite{Phelps1972,AignerCain1977}. Specifically, we adopt the model of \textcite{ChambersEchenique2021}, and enrich it slightly by allowing for prior misperceptions.%
\footnote{Chambers and Echenique's results are briefly discussed in \cref{sec:disc:accurate} below.}

The surplus generated by a worker at a firm depends on her abilities and on the task to which she is assigned. Formally, each worker has a \emph{skill type} $\theta$ drawn from a finite set $\Theta \subset \mathbb{R}$ with $\lvert \Theta \rvert \geq 2$. A \emph{task} is a vector $a \in \mathbb{R}^\Theta$, where the interpretation is that $a(\theta)\in\mathbb{R}$ is how much surplus is generated (in expectation) when a worker of skill type $\theta \in \Theta$ performs task $a$. The set of all tasks is denoted $\mathcal{A} \coloneqq \mathbb{R}^\Theta$.

Firms are described by their technology, meaning what tasks are available. Formally, a \emph{firm} is a non-empty finite set $A \subset \mathcal{A}$. A firm $A$ is called \emph{monotone} if $A\subset\mathcal{A}_M\coloneqq\{ a\in\mathbb{R}^\Theta : \text{$a(\theta')>a(\theta)$ whenever $\theta'>\theta$}\}$, meaning that each task's surplus is increasing in the worker's skill. Considering all firms $A \subset \mathcal{A}$ allows for arbitrary differentiation across skill types, whereas considering only monotone firms $A \subset \mathcal{A}_M$ means that skill types are vertically differentiated (the higher a worker's skill type, the more productive she is at \emph{every} task). We primarily focus on monotone firms.

Worker skill is unobservable, so firms must estimate skill based on observables. We describe a worker's observables by a \emph{signal} $s$. The signal should be thought of as a vector of observable characteristics, perhaps including a CV and standardized test scores. Formally, the worker population's \emph{signal structure} is a pair $\langle S, \pi \rangle $, where $S$ is a non-empty finite set and $\pi : S \times \Theta \to [0,1]$ satisfies $\sum_{s \in S} \pi(s|\theta) = 1$ for each skill type $\theta \in \Theta$.%
\footnote{Signal structures are also called ``information structures'' or ``(Blackwell) experiments.''}
The interpretation is that $S$ is the set of possible signals $s$ (e.g. possible combinations of CV contents and test scores), and that $\pi(s|\theta)$ is the probability that a type-$\theta$ worker would have signal $s$. We assume that for every signal $s \in S$, there is at least one skill type $\theta \in \Theta$ such that $\pi(s|\theta)>0$.%
\footnote{This assumption is without loss of generality, because if there were a signal $s \in S$ with $\pi(s|\theta)=0$ for every $\theta \in \Theta$, then we could neglect $s$ entirely, by deleting it from $S$.}

Firms may misperceive the distribution of skill types in the population: firms' subjective prior \emph{perception} $q \in \Delta(\Theta)$ of the skill type distribution may differ from the true distribution $p \in \Delta(\Theta)$. For simplicity, we assume that both $q$ and $p$ have full support.

To estimate a worker's ability based on her signal $s \in S$, firms apply Bayes's rule, informed by their prior perception $q \in \Delta(\Theta)$ and by the signal structure $\langle S, \pi \rangle $. Concretely, the posterior probability which firms assign to a worker with signal $s \in S$ having skill type $\theta \in \Theta$ is
$$q_{\langle S, \pi \rangle }(\theta|s) \coloneqq \frac{q(\theta) \pi(s|\theta)}{\sum_{\theta' \in \Theta} q(\theta') \pi(s|\theta')}.$$

Each firm $A \subset \mathcal{A}$ produces efficiently, allocating each worker to whichever task $a \in A$ yields the highest expected surplus given that worker's signal $s \in S$, where the expected-surplus calculations are based on the perception $q$ and the signal structure $\langle S, \pi \rangle $. Thus the maximized expected surplus calculated by a firm $A \subset \mathcal{A}$ faced with a worker with signal $s \in S$ is
\begin{equation*}
w_A(s,q,\langle S, \pi \rangle )
\coloneqq \max_{a \in A} \sum_{\theta \in \Theta} q_{\langle S, \pi \rangle }(\theta|s) a(\theta) .
\end{equation*}

Labor demand is competitive, so each worker is paid the expected surplus that she generates.%
\footnote{Our conclusions do not change if workers are instead paid a fixed (i.e. signal-independent) fraction $\alpha \in (0,1)$ of their expected surplus.}$^,$%
\footnote{As usual, the assumption that labor demand is competitive admits a game-theoretic micro-foundation à la Bertrand. It is important to note that this micro-foundation involves assuming that all firms share the same perception.}
Average pay in the population is therefore
\begin{equation*}
W_A(p,q,\langle S, \pi \rangle ) \coloneqq \sum_{\theta\in\Theta} p(\theta)
\sum_{s\in S} \pi(s|\theta) w_A(s,q,\langle S, \pi \rangle ) .
\end{equation*}
Note the differing roles of $p$ and $q$: a worker's pay depends on firms' prior perception $q$ (via the posterior perception $q_{\langle S, \pi \rangle })$, but the averaging of different workers' pay is according to the true distribution $p$.

To analyze discrimination, we compare average pay across two separate populations $I$ and $J$ with respective (true) skill distributions $p_I$ and $p_J$, perceptions $q_I$ and $q_J$, and signal structures $\langle S_I, \pi_I \rangle $ and $\langle S_J, \pi_J \rangle $. Each worker's population membership ($I$ or $J$) is observable. Our focus is on statistical discrimination in the labor market: how the gap in average pay between the two populations is influenced by the perceptions $q_I,q_J$ and signal structures $\langle S_I, \pi_I \rangle ,\langle S_J, \pi_J \rangle $. In order to isolate statistical discrimination in the labor market from the important but separate issue of human-capital inequalities arising earlier in life, we always consider populations with equal skill distributions: $p_I = p_J = p$.%
\footnote{For the same reason, we rule out taste-based discrimination; this is implicit in our assumption that workers (in any population) are paid their expected surplus.}

We compare any two perceptions $q$ and $q'$ in terms of their \emph{favorableness} in the likelihood-ratio order: $q'\succsim_{LR}q$ if $q(\theta)q'(\theta') \geq q(\theta')q'(\theta)$ holds whenever $\theta'>\theta$. We say that the population is \emph{under-perceived} if $p \succsim_{LR} q$, and \emph{over-perceived} if $q \succsim_{LR} p$.

We compare any two signal structures $\langle S, \pi \rangle $ and $\langle S', \pi' \rangle $ in terms of their \emph{informativeness} in the garbling sense \parencite{Blackwell1951,Blackwell1953}: $\langle S', \pi' \rangle $ is more informative than $\langle S, \pi \rangle $, written $\langle S', \pi' \rangle \succsim_G\langle S, \pi \rangle $, if there exists a garbling kernel from $\langle S', \pi' \rangle $ to $\langle S, \pi \rangle $, i.e. a map $g : S \times S' \to [0,1]$ satisfying $\sum_{s \in S} g(s|s') = 1$ for each $s' \in S'$ such that $\pi(s|\theta) = \sum_{s' \in S'} g(s|s') \pi'(s'|\theta)$ for each $s \in S$ and $\theta \in \Theta$. A signal structure $\langle S, \pi \rangle $ is said to be \emph{MLR (monotone likelihood ratio)} if $S \subset \mathbb{R}$ and $\pi(s|\theta) \pi(s'|\theta') \geq \pi(s|\theta') \pi(s'|\theta)$ holds whenever $s'>s$ and $\theta'>\theta$.

\begin{remark}
\label{re:unit}
An alternative interpretation of our model is that $q$ is the true distribution of skill in the population, which firms perceive correctly, and that $p$ is the distribution of skill in a sub-population of interest (or a re-weighting of the full population according to some welfare weights%
\footnote{As in \textcite{Bergson1938,Samuelson1947}, and more recently in e.g. \textcite{DworczakKominersAkbarpour2021}.}%
). On this interpretation, when comparing two populations~$I$ and $J$ with different skill distributions $q_I \neq q_J$, the comparison is made ``like-for-like,'' between two sub-populations with equal skill distributions $p_I=p_J=p$. 

Relatedly, in correspondence studies such as \textcite{BertrandMullainathan2004},%
\footnote{There is a large literature on audit/correspondence studies of discrimination, beginning with \textcite{SchwartzSkolnick1962,Daniel1968,JowellPrescottclarke1970}. See \textcite[][section~7]{Neumark2018} for a survey.}
the researcher selects the distribution of CVs that are sent to firms: in other words, the distribution of signals. This is similar in spirit to selecting a skill distribution $p$, corresponding to skills in a sub-population of interest (given a signal structure $\langle S, \pi \rangle $, this gives rise to a distribution of signals, namely $s \mapsto \sum_{\theta \in \Theta} p(\theta) \pi(s|\theta)$). Correspondence studies measure discrimination via a ``like-for-like'' comparison of average outcomes in two (sub-)populations with equal CV distributions; we similarly measure discrimination by comparing (sub-)populations with equal skill distributions.
\end{remark}

\subsection{A preliminary result}
\label{sec:model:lemma}

Intuition suggests that whatever the true skill distribution $p$ and the signal structure $\langle S, \pi \rangle $, a population's average pay will be higher the more favorably it is perceived ex ante. The following lemma formalizes this intuition, and furnishes a converse. We use this result throughout the paper.

\begin{lemma}
\label{le:favor}
Fix a skill distribution $p$, and consider two perceptions $q$ and $q'$.
\begin{enumerate}[label=(\alph*)]
\item \label{le:favor:suff} If $q'$ is more favorable than $q$ ($q'\succsim_{LR}q$), then for any signal structure $\langle S, \pi \rangle $, 
$$W_A(p,q',\langle S, \pi \rangle )\geq W_A(p,q,\langle S, \pi \rangle ) \quad \text{for any monotone firm $A\subset \mathcal{A}_M$.}$$
\item \label{le:favor:nec} If $q'$ is not more favorable than $q$ ($q'\not\succsim_{LR}q$), then there exists a signal structure $\langle S, \pi \rangle $ such that 
$$W_A(p,q',\langle S, \pi \rangle ) < W_A(p,q,\langle S, \pi \rangle ) \quad \text{for any monotone firm $A\subset \mathcal{A}_M$.}$$
\end{enumerate}
\end{lemma}

\section{Decomposing the impact of new information}
\label{sec:decomp}

In this section, we present our main theoretical contribution: a decomposition of the impact of new information on average pay into two components reflecting distinct \emph{instrumental} and \emph{perception-correcting} effects, and a result which signs these two components. The instrumental effect captures how firms use extra information to tailor task assignment more finely to workers' likely skills. The perception-correcting effect reflects the diminished importance of any misperception $q \neq p$ as information becomes more precise; indeed, in the extreme case in which the signal structure becomes fully informative, the perception $q$ ceases entirely to matter for pay, since each worker's skill type is revealed and she is paid accordingly.

To describe our decomposition, fix a firm $A \subset \mathcal{A}$, a skill distribution $p$, and a perception $q$. Consider a shift from signal structure $\langle S, \pi \rangle $ to $\langle S', \pi' \rangle $, where the latter is more informative (that is, $\langle S', \pi' \rangle  \succsim_G \langle S, \pi \rangle $), and let $g$ be a garbling kernel from $\langle S', \pi' \rangle $ to $\langle S, \pi \rangle $.%
\footnote{\label{footnote:selection_g}The garbling kernel $g$ is unique if $\pi'$, when viewed as an $\lvert S' \rvert \times \lvert \Theta \rvert$ matrix, has rank equal to $\lvert S' \rvert$. This condition holds generically if $\lvert S' \rvert \leq \lvert \Theta \rvert$, and fails if $\lvert S' \rvert > \lvert \Theta \rvert$. If the condition fails, then there are multiple garbling kernels from $\langle S', \pi' \rangle $ to $\langle S, \pi \rangle $; we arbitrarily select one, $g$, and hold it fixed throughout. If a different selection were made, then the terms of our decomposition would potentially be different; however, our results about these terms are valid whatever the selection.}
Then there exists a joint distribution of signals and skill such that the distribution of $\theta$ conditional on both $s$ and $s'$ is the same as that conditional on $s'$ alone (in other words, $\theta$ is independent of $s$ conditional on $s'$). Specifically, that joint distribution is $\mu_p \in \Delta(\Theta \times S \times S')$ given by
$$\mu_p(\theta,s,s') \coloneqq p(\theta) \pi'(s'|\theta) g(s|s')
\quad \text{for all $\theta \in \Theta$, $s \in S$ and $s' \in S'$.}$$
We denote the derived marginal distributions by $\mu_p(s)$, $\mu_p(s')$, $\mu_p(s,s')$, and so on, and the conditional distributions by $\mu_p(\theta|s)$, $\mu_p(\theta|s')$, $\mu_p(\theta|s,s')$, $\mu_p(s'|s)$, etc. From the perspective of firms, who believe that skill is distributed according to $q$ rather than $p$, the corresponding joint distribution is $\mu_q \in \Delta(\Theta \times S \times S')$ given by
$$\mu_q(\theta,s,s') \coloneqq q(\theta) \pi'(s'|\theta) g(s|s')
\quad \text{for all $\theta \in \Theta$, $s \in S$ and $s' \in S'$.}$$
For each signal $s \in S$ of the less informative signal structure $\langle S, \pi \rangle $, let
$$\widehat{a}_s \in \argmax_{a \in A} \sum_{\theta \in \Theta} q_{\langle S, \pi \rangle }(\theta|s) a(\theta)$$
denote firm $A$'s surplus-maximizing task choice for a worker with signal $s$.%
\footnote{\label{footnote:selection_a}In the (non-generic) case in which there are multiple surplus-maximizing task assignments, we arbitrarily select one, $s \mapsto \widehat{a}_s$, and hold it fixed throughout. As in \cref{footnote:selection_g}, the terms of our decomposition are not invariant to which selection is made, but our results hold for any such selection.}
Similarly, for each $s' \in S'$, let
$$\widehat{a}'_{s'} \in \argmax_{a \in A} \sum_{\theta \in \Theta} q_{\langle S', \pi' \rangle }(\theta|s') a(\theta)$$
denote a surplus-maximizing task choice under the more informative signal structure $\langle S', \pi' \rangle $.%
\footnote{Again, we arbitrarily select one task assignment, $s' \mapsto \widehat{a}'_{s'}$, and hold it fixed throughout.}

We shall decompose the change $W_A(p,q,\langle S', \pi' \rangle ) - W_A(p,q,\langle S, \pi \rangle )$ in average pay into two terms. The first term, which we call \emph{perception-correcting,} captures the effect on average pay of the change in beliefs induced by the extra information contained in $\langle S', \pi' \rangle $, \emph{holding fixed} the firm's task-assignment choices $s \mapsto \widehat{a}_s$. The second term, called \emph{instrumental,} captures the increase of expected surplus from firms using the extra information in $\langle S', \pi' \rangle $ to better tailor task assignment to workers' likely skills, \emph{holding perceptions fixed.} The formal definitions of these two effects are as follows.

\begin{definition}
For any firm $A \subset \mathcal{A}$, skill distribution $p$, perception $q$, and signal structures $\langle S', \pi' \rangle  \succsim_G \langle S, \pi \rangle $, the \emph{perception-correcting component} of the change in average pay is
\begin{multline*}
\mathcal{C}_A(p,q,\langle S, \pi \rangle ,\langle S', \pi' \rangle )
\\
\begin{aligned}
\coloneqq{}& 
\sum_{s' \in S'} \mu_p(s') \sum_{s \in S} 
\left[ 1 - \frac{ \mu_p(s) / \mu_q(s) }{ \mu_p(s') / \mu_q(s') } \right]
\sum_{\theta \in \Theta} q_{\langle S', \pi' \rangle }(\theta|s') g(s|s') \widehat{a}_s(\theta)
\\
={}& 
\sum_{s \in S} \mu_p(s)
\sum_{s' \in S'} \left[ \mu_p(s'|s) - \mu_q(s'|s) \right]
\sum_{\theta \in \Theta} \mu_q(\theta|s,s') \widehat{a}_s(\theta) ,
\end{aligned}
\end{multline*}
and the \emph{instrumental component} of the change in average pay is
\begin{multline*}
\mathcal{I}_A(p,q,\langle S, \pi \rangle ,\langle S', \pi' \rangle )
\\
\begin{aligned}
\coloneqq{}& 
 \sum_{s \in S} \sum_{s' \in S'} \mu_p(s,s') 
\sum_{\theta \in \Theta} \mu_q(\theta|s,s') \left[ \widehat{a}'_{s'}(\theta) - \widehat{a}_s(\theta) \right]
\\
={}&
\sum_{s' \in S'} \mu_p(s')
\sum_{\theta \in \Theta} q_{\langle S', \pi' \rangle }(\theta|s')
\left[ \widehat{a}'_{s'}(\theta)
- \sum_{s \in S} g(s|s') \widehat{a}_s(\theta) \right] .
\end{aligned}
\end{multline*}
\end{definition}

\begin{example}
\label{ex:decomp}
Suppose that skill types are binary, $\Theta = \{0,1\}$. Let $\langle S, \pi \rangle$ be an uninformative signal structure: $\pi(s|0)=\pi(s|1)$ for every $s \in S$. Let $\langle S', \pi' \rangle$ be fully informative: $S' = \{s^0,s^1\}$ and $\pi'(s^0|0) = \pi'(s^1|1) = 1$. Obviously these signal structures are MLR and satisfy $\langle S', \pi' \rangle \succ_G \langle S, \pi \rangle$.

Let $\overline{a},\widetilde{a} \in \mathcal{A}_M$ be the tasks defined by $\overline{a}(\theta)\coloneqq \theta$ and $\widetilde{a}(\theta)\coloneqq \theta/2 + 1/4$ for each $\theta \in \Theta$. Consider the monotone firm $A=\{\overline{a},\widetilde{a}\}$, and suppose to begin with that the population is under-perceived: $q(1)=1/4$ and $p(1) = 3/4$.

Under the uninformative signal structure $\langle S,\pi\rangle$, each worker is believed to be the high skill type $\theta=1$ with probability $q(1)=1/4$, so tasks $\overline{a}$ and $\widetilde{a}$ yield expected surplus $(3/4) \cdot 0 + (1/4) \cdot 1 = 1/4$ and $(3/4) \cdot (1/4) + (1/4) \cdot (3/4) = 3/8$, respectively. Hence every worker is assigned to $\widetilde{a}$ and paid $3/8$.

The perception-correcting component captures how the extra information in the fully informative signal structure $\langle S', \pi' \rangle $ affects average pay, holding task assignment fixed: that is, we imagine that all workers continue to be assigned to task $\widetilde{a}$. The fully informative signal structure $\langle S', \pi' \rangle $ reveals each worker's skill type $\theta$, so the firm learns that the true fraction of high skill types is $p(1)=3/4$ rather than $q(1)=1/4$. The resulting change in average pay is
\begin{equation*}
\mathcal{C}_A (p,q,\langle S,\pi\rangle,\langle S',\pi'\rangle)
= \left( \frac{1}{4}-\frac{3}{4} \right) \cdot \widetilde{a}(0)
+ \left( \frac{3}{4}-\frac{1}{4} \right) \cdot \widetilde{a}(1)
= \frac{1}{4},
\end{equation*}
as represented by the arrow in \Cref{fig:2-act-under:correction}.

\begin{figure}
\centering
\begin{subfigure}[t]{0.48\textwidth}
\centering
\begin{tikzpicture}[x=5cm, y=5cm, line cap=round]

  \draw (0,0) -- (1,0);
  \draw[->] (0,0) -- (0,1.08);

  \node at (0,-0.06) {$0$};
  \draw (1,-0.01) -- (1,0.01);
  \node at (1,-0.06) {$1$};

  \node[left] at (0,0.25) {$\tfrac{1}{4}$};

  \draw (-0.01,0.375) -- (0.01,0.375);
  \node[left] at (0,0.375) {$\tfrac{3}{8}$};

  \draw (-0.01,0.625) -- (0.01,0.625);
  \node[left] at (0,0.625) {$\tfrac{5}{8}$};

  \draw (-0.01,0.75) -- (0.01,0.75);
  \node[left] at (0,0.75) {$\tfrac{3}{4}$};

  \draw[thick] (0,0) -- (1,1);
  \node[right] at (1,1) {$\overline{a}$};

  \draw[thick] (0,0.25) -- (1,0.75);
  \node[right] at (1,0.75) {$\widetilde{a}$};

  \draw (0.25,-0.01) -- (0.25,0.01);
  \node at (0.25,-0.06) {$q(1)$};

  \draw (0.75,-0.01) -- (0.75,0.01);
  \node at (0.75,-0.06) {$p(1)$};

  \draw[thin,dotted] (0.25,0) -- (0.25,0.375) -- (0.75,0.375) -- (0.75,0);
  \fill (0.25,0.375) circle (1.25pt);

  \fill (0.75,0.625) circle (1.25pt);
  \draw[very thick,->] (0.75,0.375) -- (0.75,0.62);
  \node[right] at (0.75,0.5) {$\mathcal{C}$};
\end{tikzpicture}
\caption{Perception-correcting term $\mathcal{C}$.}
\label{fig:2-act-under:correction}
\end{subfigure}
\hfill
\begin{subfigure}[t]{0.48\textwidth}
\centering
\begin{tikzpicture}[x=5cm, y=5cm, line cap=round]

  \draw (0,0) -- (1,0);
  \draw[->] (0,0) -- (0,1.08);

  \node at (0,-0.06) {$0$};
  \draw (1,-0.01) -- (1,0.01);
  \node at (1,-0.06) {$1$};

  \draw (-0.01,0.625) -- (0.01,0.625);
  \node[left] at (0,0.625) {$\tfrac{5}{8}$};

  \draw (-0.01,0.8125) -- (0.01,0.8125);
  \node[left] at (0,0.8125) {$\tfrac{13}{16}$};

  \node[left] at (0,0.25) {$\tfrac{1}{4}$};


  \draw (-0.01,1) -- (0.01,1);
  \node[left] at (0,1) {$1$};

  \draw[thick] (0,0) -- (1,1);
  \node[right] at (1,1) {$\overline{a}$};

  \draw[thick] (0,0.25) -- (1,0.75);
  \node[right] at (1,0.75) {$\widetilde{a}$};

  \draw (0.25,-0.01) -- (0.25,0.01);
  \node at (0.25,-0.06) {$q(1)$};

  \draw (0.75,-0.01) -- (0.75,0.01);
  \node at (0.75,-0.06) {$p(1)$};

  \draw[thin,dotted] (0,0.25) -- (1,1);

  \draw[thin,dotted,opacity=0.25] (0.25,0) -- (0.25,0.375) -- (0.75,0.375) -- (0.75,0);
  \fill (0.25,0.375) circle (1.25pt);

  \fill (0.75,0.625) circle (1.25pt);
  \begin{scope}[transparency group,opacity=0.25]
    \draw[very thick,->] (0.75,0.375) -- (0.75,0.62);
  \end{scope}
  \node[right,opacity=0.25] at (0.75,0.5) {$\mathcal{C}$};

  \fill (0.75,0.8125) circle (1.25pt);
  \draw[very thick,->] (0.75,0.625) -- (0.75,0.8075);
  \node[right] at (0.75,0.72) {$\mathcal{I}$};
  
\end{tikzpicture}
\caption{Instrumental term $\mathcal{I}$.}
\label{fig:2-act-under:instrumental}
\end{subfigure}
\caption{Decomposition in \Cref{ex:decomp} with under-perception.}
\label{fig:2-act-under}
\end{figure}

The instrumental component captures how average pay is affected by the change in task assignment caused by the availability of extra information, holding perceptions fixed at their new (``corrected'') level: that is, we imagine that the firm (correctly) believes that the fraction of type $\theta=1$ workers is $p(1)=3/4$.%
\footnote{In detail, applying the definitions, we see that what is ``corrected'' in the perception-correcting term is the \emph{signal-conditional-on-signal} distribution, which shifts from $\mu_q(s'|s)$ to $\mu_p(s'|s)$, and is then held fixed at $\mu_p(s'|s)$ in the instrumental term. (The type-conditional-on-signals distribution is the same in both terms, namely $\mu_q(\theta|s,s')$.) The present example has the simplifying feature that for each $r \in \{p,q\}$, the distribution $\mu_r(s'|s)$ is effectively the same as $r$ itself, since $\langle S, \pi \rangle $ is uninformative, $S'=\{s^0,s^1\}$, and $\pi'(s^\theta|\theta)=1$ for each $\theta \in \Theta$, so that $\mu_r(s^\theta|s) = \mu_r(s^\theta) = r(\theta)$ for all $s \in S$ and $\theta \in \Theta$.}
Since the extra information fully reveals workers' skill types and $\overline{a}(0)=0<1/4=\widetilde{a}(0)$ and $\overline{a}(1)=1>3/4=\widetilde{a}(1)$, workers of type $\theta=0$ continue to be assigned to task $\widetilde{a}$, while those of type $\theta=1$ are re-assigned to task $\overline{a}$. This changes average pay by
\begin{equation*}
\mathcal{I}_A (p,q,\langle S,\pi\rangle,\langle S',\pi'\rangle)
= p(0) \cdot 0
+ p(1) \cdot \left(1-\frac{3}{4}\right)
= \frac{3}{16},
\end{equation*}
as represented by the black arrow in \Cref{fig:2-act-under:instrumental}.

In case the population is over-perceived, $q(1)=3/4$ and $p(1) = 1/4$, symmetric calculations reveal that the instrumental term remains positive, but the perception-correcting term is negative, as depicted in \Cref{fig:2-act-under-rev}.

\begin{figure}[ht]
\centering
\begin{subfigure}[t]{0.48\textwidth}
\centering
\begin{tikzpicture}[x=5cm, y=5cm, line cap=round]

  \draw (0,0) -- (1,0);
  \draw[->] (0,0) -- (0,1.08);

  \node at (0,-0.06) {$0$};

  \draw (0.25,-0.01) -- (0.25,0.01);
  \node at (0.25,-0.06) {$p(1)$};

  \draw (0.75,-0.01) -- (0.75,0.01);
  \node at (0.75,-0.06) {$q(1)$};

  \draw (1,-0.01) -- (1,0.01);
  \node at (1,-0.06) {$1$};

  \node[left] at (-0.015,0.25) {$\frac{1}{4}$};

  \draw (-0.01,0.75) -- (0.01,0.75);
  \node[left] at (-0.015,0.75) {$\frac{3}{4}$};

  \draw (-0.01,1) -- (0.01,1);
  \node[left] at (-0.015,1) {$1$};

  \draw[thick] (0,0) -- (1,1);
  \node[right] at (1,1) {$\overline{a}$};

  \draw[thick] (0,0.25) -- (1,0.75);
  \node[right] at (1,0.75) {$\widetilde{a}$};

  \draw[thin,dotted] (0.75,0) -- (0.75,0.75) -- (0.25,0.75) -- (0.25,0);
  \fill (0.75,0.75) circle (1.25pt);

  \fill (0.25,0.25) circle (1.25pt);
  \draw[very thick,->] (0.25,0.75) -- (0.25,0.255);
  \node[right] at (0.25,0.6) {$\mathcal{C}$};

\end{tikzpicture}
\caption{Perception-correcting term $\mathcal{C}$.}
\label{fig:2-act-under:correction-rev}
\end{subfigure}
\hfill
\begin{subfigure}[t]{0.48\textwidth}
\centering
\begin{tikzpicture}[x=5cm, y=5cm, line cap=round]

  \draw (0,0) -- (1,0);
  \draw[->] (0,0) -- (0,1.08);

  \node at (0,-0.06) {$0$};

  \draw (0.25,-0.01) -- (0.25,0.01);
  \node at (0.25,-0.06) {$p(1)$};

  \draw (0.75,-0.01) -- (0.75,0.01);
  \node at (0.75,-0.06) {$q(1)$};

  \draw (1,-0.01) -- (1,0.01);
  \node at (1,-0.06) {$1$};

  \node[left] at (-0.015,0.25) {$\frac{1}{4}$};

  \draw (-0.01,0.4375) -- (0.01,0.4375);
  \node[left] at (-0.015,0.4375) {$\frac{7}{16}$};

  \draw (-0.01,0.75) -- (0.01,0.75);
  \node[left] at (-0.015,0.75) {$\frac{3}{4}$};

  \draw (-0.01,1) -- (0.01,1);
  \node[left] at (-0.015,1) {$1$};

  \draw[thick] (0,0) -- (1,1);
  \node[right] at (1,1) {$\overline{a}$};

  \draw[thick] (0,0.25) -- (1,0.75);
  \node[right] at (1,0.75) {$\widetilde{a}$};

  \draw[thin,dotted] (0,0.25) -- (1,1);

  \fill (0.75,0.75) circle (1.25pt);

  \fill (0.25,0.25) circle (1.25pt);
  \begin{scope}[transparency group,opacity=0.25]
    \draw[thin,dotted] (0.75,0) -- (0.75,0.75) -- (0.25,0.75) -- (0.25,0);
    \draw[very thick,->] (0.25,0.75) -- (0.25,0.255);
  \end{scope}
  \node[right,opacity=0.25] at (0.25,0.6) {$\mathcal{C}$};

  \fill (0.25,0.4375) circle (1.25pt);
  \draw[very thick,->] (0.25,0.25) to[out=160,in=200] (0.25,0.4375);
  \node[left] at (0.21,0.27) {$\mathcal{I}$};

\end{tikzpicture}
\caption{Instrumental term $\mathcal{I}$.}
\label{fig:2-act-under:instrumental-rev}
\end{subfigure}
\caption{Decomposition in \Cref{ex:decomp} with over-perception.}
\label{fig:2-act-under-rev}
\end{figure}
\end{example}

Our main result asserts that the change in average pay decomposes into perception-correcting and instrumental components, and that these components may be signed.

\begin{theorem}
\label{th:decomp}
Fix a firm $A \subset \mathcal{A}$, a skill distribution $p$, a perception $q$, and signal structures $\langle S', \pi' \rangle  \succsim_G \langle S, \pi \rangle $.
\begin{enumerate}[label=(\alph*)]
\item \label{th:decomp:decomp} The change in average pay admits the decomposition
\begin{multline*}
W_A(p,q,\langle S', \pi' \rangle ) - W_A(p,q,\langle S, \pi \rangle )
\\
= \mathcal{C}_A(p,q,\langle S, \pi \rangle ,\langle S', \pi' \rangle )
+ \mathcal{I}_A(p,q,\langle S, \pi \rangle ,\langle S', \pi' \rangle ) .
\end{multline*}
\item \label{th:decomp:instr} The instrumental component $\mathcal{I}_A(p,q,\langle S, \pi \rangle ,\langle S', \pi' \rangle )$ is non-negative.
\end{enumerate}
Suppose in addition that the firm is monotone ($A \subset \mathcal{A}_M$) and that the more informative signal structure $\langle S', \pi' \rangle $ is MLR.
\begin{enumerate}[label=(\alph*),resume]
\item \label{th:decomp:corr_pos} If the population is under-perceived ($p \succsim_{LR} q$), then the perception-correcting component $\mathcal{C}_A(p,q,\langle S, \pi \rangle ,\langle S', \pi' \rangle )$ is non-negative.
\item \label{th:decomp:corr_neg} If the population is over-perceived ($q \succsim_{LR} p$), then the perception-correcting component $\mathcal{C}_A(p,q,\langle S, \pi \rangle ,\langle S', \pi' \rangle )$ is non-positive.
\end{enumerate}
\end{theorem}

The non-negativity of the instrumental component (part~\ref{th:decomp:instr}) is a slight generalization of the ``easy'' half of Blackwell's (\citeyear{Blackwell1951,Blackwell1953}) theorem, which asserts the same conclusion under the additional hypothesis that the perception $q$ is accurate, i.e. equal to the true distribution $p$. The familiar intuition is that extra information cannot hurt a decision-maker, since she can always ignore it. (Blackwell's theorem is stated in full in \cref{sec:disc:accurate} below.)

This familiar intuition is invalid in our model, because extra information also has the secondary effect of partially correcting any prior misperception $q \neq p$. The definition of the instrumental component $\mathcal{I}_A(p,q,\langle S, \pi \rangle ,\langle S', \pi' \rangle )$ shuts down this effect by holding perceptions fixed, thereby restoring the usual non-negativity of the instrumental value of information (part~\ref{th:decomp:instr}).

The secondary effect is isolated in the perception-correcting component, in which task assignment is held fixed. This term is zero if the prior perception $q$ is accurate, i.e. equal to the true distribution. For an under-perceived population, the perception-correcting effect is non-negative (part~\ref{th:decomp:corr_pos}) because as more information becomes available, firms recognize that the skill distribution is better than their prior perception $q$ suggested, leading them to pay workers more. The same logic explains why the perception-correcting term is non-positive for an over-perceived population (part~\ref{th:decomp:corr_neg}). These intuitions are incomplete because they do not exploit the theorem's monotonicity and MLR assumptions, without which the result may fail, as we show in \cref{sec:decomp:mlr_mon} below.

More abstractly, the force behind the perception-correcting term is that posterior beliefs about workers' skill are (bary)centered somewhere between $q$ and $p$, and that this point is further from $q$ and closer to $p$ the more informative is the signal structure (it is $q$ under no information, and $p$ under full information). This reflects the fact that Bayesians are not dogmatic, but rather update beliefs about a population in light of the evidence.

The idea of a perception-correcting effect has appeared in various contexts, such as statistics,%
\footnote{See \textcite{Laplace1774,Wald1947,Girsanov1960,BlackwellDubins1962}.}
persuasion with heterogeneous priors \parencite[e.g.][]{AlonsoCamara2016,OnuchicRay2023}, and behavioral decision theory \parencite{Bordoli2025}. Bordoli's paper in particular distinguishes between ``instrumental'' and ``non-instrumental'' effects of information. Closest to \Cref{th:decomp}\ref{th:decomp:corr_pos}--\ref{th:decomp:corr_neg} is the result of \textcite{KartikLeeSuen2021}, which formalizes a sense in which posterior beliefs become (bary)centered further from $q$ and closer to $p$ as the signal structure becomes more informative. This result formally coincides with the special case of \Cref{th:decomp}\ref{th:decomp:corr_pos}--\ref{th:decomp:corr_neg} in which only single-task firms $A=\{a\}$ are considered, and the proofs are accordingly similar.

\subsection{Signing the total impact of new information}
\label{sec:decomp:total}

\Cref{th:decomp} can sometimes be used to determine the sign of the total impact $W_A(p,q,\langle S', \pi' \rangle ) - W_A(p,q,\langle S, \pi \rangle )$ of new information on average pay.

\begin{corollary}
\label{cor:total}
Fix a skill distribution $p$, a perception $q$, and signal structures $\langle S', \pi' \rangle  \succsim_G \langle S, \pi \rangle $, where $\langle S', \pi' \rangle$ is MLR. If the population is under-perceived ($p \succsim_{LR} q$), then new information increases average pay:
$$W_A(p,q,\langle S', \pi' \rangle )\geq W_A(p,q,\langle S, \pi \rangle )
\quad \text{for any monotone firm $A\subset \mathcal{A}_M$}.$$
\end{corollary}

\begin{namedthm}[\Cref*{ex:decomp} {\normalfont (continued)}.]
\label{ex:decomp_disc}
As seen in \Cref{fig:2-act-under}, when the population is under-perceived, both the instrumental and perception-correcting terms are positive, so extra information increases average pay.
\end{namedthm}

Conversely, if the population is over-perceived ($q \succsim_{LR} p$), so that the instrumental and perception-correcting components have opposite signs, then the total effect is ambiguous and can be negative, as in \Cref{fig:2-act-under-rev} and in the following example. (This example was discussed in the introduction.)

\begin{example}
\label{ex:reversal}
Consider \Cref{ex:decomp}, except with the single-task (monotone) firm $A = \{\overline{a}\}$. This firm pays each worker their posterior probability of being the high skill type: $w_{\left\{ \overline{a} \right\}}(s'',q,\langle S'', \pi'' \rangle ) = q_{\langle S'', \pi'' \rangle }(1|s'')$ for any signal structure $\langle S'', \pi'' \rangle $ and signal $s'' \in S''$. Hence the change in average pay is
\begin{multline*}
W_{\left\{ \overline{a} \right\}}(p,q,\langle S', \pi' \rangle )
- W_{\left\{ \overline{a} \right\}}(p,q,\langle S, \pi \rangle )
\\
= \left[ p(0) q_{\langle S', \pi' \rangle }\left(1\middle|s^0\right)
+ p(1) q_{\langle S', \pi' \rangle }\left(1\middle|s^1\right) \right]
- q(1)
= p(1) - q(1) .
\end{multline*}
Thus if the population is strictly over-perceived ($q(1) > p(1)$), then adding information strictly \emph{decreases} average pay.

By \Cref{th:decomp}, such a ``reversal'' can occur only if the perception-correcting component is sufficiently negative to dominate the always non-negative instrumental component. To verify this directly, observe that the instrumental component is zero, since that term captures the use of new information to change task assignment, and there is only one task available:
\begin{multline*}
\mathcal{I}_{\left\{ \overline{a} \right\}}(p,q,\langle S, \pi \rangle ,\langle S', \pi' \rangle )
\\
= \sum_{s \in S} \mu_p(s)
\sum_{s' \in S'} \mu_p(s'|s)
\sum_{\theta \in \Theta} \mu_{q}(\theta|s,s') \left[ \overline{a}(\theta) - \overline{a}(\theta) \right]
= 0 .
\end{multline*}
The perception-correcting term, on the other hand, is equal to
\begin{multline*}
\mathcal{C}_{\left\{ \overline{a} \right\}}(p,q,\langle S, \pi \rangle ,\langle S', \pi' \rangle )
\\
\begin{aligned}
&= \sum_{s \in S} \mu_p(s)
\sum_{s' \in S'} \left[ \mu_p(s'|s) - \mu_{q}(s'|s) \right]
\sum_{\theta \in \Theta} \mu_{q}(\theta|s,s') \overline{a}(\theta)
\\
&= \sum_{\theta' \in \Theta} \left[ \mu_p\left(s^{\theta'}\right) - \mu_{q}\left(s^{\theta'}\right) \right]
\sum_{\theta \in \Theta} \mu_{q}\left(\theta\middle|s^{\theta'}\right) \overline{a}(\theta)
\\
&= \sum_{\theta' \in \Theta} \left[ p(\theta') - q(\theta') \right]
\overline{a}(\theta') 
= p(1) - q(1) .
\end{aligned}
\end{multline*}
\end{example}

\subsection{The role of the monotonicity and MLR assumptions}
\label{sec:decomp:mlr_mon}

In \Cref{th:decomp}\ref{th:decomp:corr_pos}--\ref{th:decomp:corr_neg}, it is assumed that the firm under consideration is monotone and that the more informative of the two signal structures is MLR. The following two examples illustrate the role of these assumptions by showing how without them, the conclusions of \Cref{th:decomp}\ref{th:decomp:corr_pos}--\ref{th:decomp:corr_neg} may fail.

\begin{example}
\label{ex:mon_fail}
Consider \Cref{ex:reversal}, except with a different firm: let $\underline{a} \in \mathcal{A} \setminus \mathcal{A}_M$ be the task given by $\underline{a}(\theta) \coloneqq 1-\theta$ for each $\theta \in \Theta$, and consider the non-monotone firm $A = \{\underline{a}\}$.%
\footnote{The firm~$\left\{\underline{a}\right\}$ is not monotone since by definition, a ``monotone'' firm~$A$ is one whose every task $a \in A$ is strictly \emph{increasing.}}
This firm pays each worker their posterior probability of being the low skill type: $w_{\left\{ \underline{a} \right\}}(s'',q,\langle S'', \pi'' \rangle ) = q_{\langle S'', \pi'' \rangle }(0|s'')$ for any signal structure $\langle S'', \pi'' \rangle $ and signal $s'' \in S''$. Hence the change in average pay is
\begin{multline*}
W_{\left\{ \underline{a} \right\}}(p,q,\langle S', \pi' \rangle )
- W_{\left\{ \underline{a} \right\}}(p,q,\langle S, \pi \rangle )
\\
= \left[ p(0) q_{\langle S', \pi' \rangle }\left(0\middle|s^0\right)
+ p(1) q_{\langle S', \pi' \rangle }\left(0\middle|s^1\right) \right]
- q(0)
= p(0) - q(0) .
\end{multline*}
The instrumental component $\mathcal{I}_A(p,q,\langle S, \pi \rangle ,\langle S', \pi' \rangle )$ is zero since task assignment does not change (as firm $\left\{ \underline{a} \right\}$ has only one task), so the perception-correcting component is $\mathcal{C}_A(p,q,\langle S, \pi \rangle ,\langle S', \pi' \rangle ) = p(0)-q(0)$ by \Cref{th:decomp}\ref{th:decomp:decomp}. Hence the perception-correcting component is strictly negative if the population is under-perceived ($q(0)>p(0)$) and strictly positive if the population is over-perceived ($q(0)<p(0)$). This shows that the monotonicity hypothesis cannot be dispensed with in \Cref{th:decomp}\ref{th:decomp:corr_pos}--\ref{th:decomp:corr_neg}.
\end{example}

\begin{example}
\label{ex:mlr_fail}
Suppose that skill types are ternary, $\Theta = \{0,1,2\}$, and consider the perception $q$ given by $q(0)=q(1)=1/4$ and $q(2)=1/2$. Let $\langle S, \pi \rangle$ be an uninformative signal structure: $\pi(s|0)=\pi(s|1)=\pi(s|2)$ for every $s \in S$. Let $\langle S', \pi' \rangle$ reveal whether skill is $1$ and nothing else: $\pi'\bigl(s^1\bigm|1\bigr) = \pi'\bigl(s^{0,2}\bigm|0\bigr) = \pi'\bigl(s^{0,2}\bigm|2\bigr) = 1$ for some $s^1 \neq s^{0,2}$ in $S'$.%
\footnote{This simple signal structure is not new; see e.g. the introduction of \textcite{AlonsoCamara2016} and the supplemental appendix of \textcite{KartikLeeSuen2021}.}
By inspection, $\langle S', \pi' \rangle \succ_G \langle S, \pi \rangle$, and the signal structure $\langle S', \pi' \rangle$ is not MLR. Posterior beliefs about skill are $q_{\langle S, \pi \rangle }(\cdot|s) = q(\cdot)$ for every $s \in S$, $q_{\langle S', \pi' \rangle }\left(1\middle|s^1\right) = 1$, $q_{\langle S', \pi' \rangle }\left(0\middle|s^{0,2}\right) = 1/3$ and $q_{\langle S', \pi' \rangle }\left(2\middle|s^{0,2}\right) = 2/3$.

Let $\overline{a} \in \mathcal{A}_M$ be the task given by $\overline{a}(\theta) \coloneqq \theta$ for each $\theta \in \Theta$, and consider the monotone firm $A = \{\overline{a}\}$. This firm pays each worker their posterior expected type: $w_{\left\{ \overline{a} \right\}}(s'',q,\langle S'', \pi'' \rangle ) = q_{\langle S'', \pi'' \rangle }(1|s'') + 2 q_{\langle S'', \pi'' \rangle }(2|s'')$ for any signal structure $\langle S'', \pi'' \rangle $ and signal $s'' \in S''$. Hence the change in average pay is
\begin{multline*}
W_{\left\{ \overline{a} \right\}}(p,q,\langle S', \pi' \rangle )
- W_{\left\{ \overline{a} \right\}}(p,q,\langle S, \pi \rangle )
\\
= \left[ p(1) \cdot 1 + (1-p(1)) \cdot \left( \frac{1}{3} \cdot 0 + \frac{2}{3} \cdot 2 \right) \right]
- \left[ \frac{1}{4} \cdot 1 + \frac{1}{2} \cdot 2 \right]
= \frac{1}{3} \left( \frac{1}{4}-p(1) \right) .
\end{multline*}
The instrumental component $\mathcal{I}_A(p,q,\langle S, \pi \rangle ,\langle S', \pi' \rangle )$ is zero since task assignment does not change (as firm $\left\{ \overline{a} \right\}$ has only one task), so the perception-correcting term is $\mathcal{C}_A(p,q,\langle S, \pi \rangle ,\langle S', \pi' \rangle ) = (1/4-p(1))/3$ by \Cref{th:decomp}\ref{th:decomp:decomp}.

Let $p(0)=1/4 - 3 \delta$, $p(1) = 1/4 + \delta$ and $p(2) = 1/2 + 2 \delta$, where $\delta \in (-1/4,1/12)$. If $\delta>0$, then the population is under-perceived ($p \succsim_{LR} q$), but $\mathcal{C}_A(p,q,\langle S, \pi \rangle ,\langle S', \pi' \rangle ) < 0$. This shows that the MLR hypothesis cannot be dispensed with in \Cref{th:decomp}\ref{th:decomp:corr_pos}. If $\delta<0$, then the population is over-perceived ($q \succsim_{LR} p$), but $\mathcal{C}_A(p,q,\langle S, \pi \rangle ,\langle S', \pi' \rangle ) > 0$, so the MLR hypothesis cannot be dispensed with in \Cref{th:decomp}\ref{th:decomp:corr_neg}, either.
\end{example}

\section{Implications for statistical discrimination}
\label{sec:disc}

In this section, we explore the implications of our decomposition result (\Cref{th:decomp}) for pay gaps between different populations. We focus on \emph{statistical discrimination,} meaning gaps in average pay between populations that have the same (true) skill distribution.

One message of \Cref{th:decomp} is that under some conditions, a more informative population will earn more on average. \Cref{le:favor} says that a more favorably perceived population is paid more on average. Combining these two insights yields the following prediction about statistical discrimination:

\begin{corollary}
\label{cor:disc}
Consider two populations $I$ and $J$, both with skill distribution~$p$ and with respective perceptions $q_I$ and $q_J$ and signal structures $\langle S_I, \pi_I \rangle $ and $\langle S_J, \pi_J \rangle $. Suppose that $\langle S_I, \pi_I \rangle $ is MLR, and that population~$J$ is under-perceived ($p \succsim_{LR} q_J$). If population~$I$ is both more favorably perceived than and more informative than population~$J$ (i.e. $q_I \succsim_{LR} q_J$ and $\langle S_I, \pi_I \rangle  \succsim_G \langle S_J, \pi_J \rangle $), then monotone firms pay population~$I$ more on average:
$$W_A(p,q_I,\langle S_I, \pi_I \rangle )\geq W_A(p,q_J,\langle S_J, \pi_J \rangle )
\quad \text{for any monotone firm $A\subset \mathcal{A}_M$}.$$
\end{corollary}

\begin{proof}
For any monotone firm $A \subset \mathcal{A}_M$, we have
\begin{equation*}
W_A(p,q_I,\langle S_I, \pi_I \rangle )
\geq W_A(p,q_J,\langle S_I, \pi_I \rangle )
\geq W_A(p,q_J,\langle S_J, \pi_J \rangle ) 
\end{equation*}
by \Cref{le:favor} (first inequality) and \Cref{cor:total} (second inequality).
\end{proof}

The key takeaway from \Cref{cor:disc} is that 
greater informativeness and more favorable perception translate straightforwardly into higher average pay \emph{provided} the disfavored population ($J$) is under-perceived ex ante. Outside of that case, a population that is more informative and more favorably perceived may be paid \emph{less} on average, as the following example shows.

\begin{namedthm}[\Cref*{ex:reversal} {\normalfont (continued)}.]
\label{ex:reversal_disc}
Consider two populations, $I$ and $J$, with equal skill type distributions $p_I = p_J = p$ and respective perceptions $q_I$ and $q_J$. Population~$I$ has the fully informative signal structure $\langle S_I, \pi_I \rangle \coloneqq \langle S', \pi' \rangle$, while population~$J$ has the totally uninformative signal structure $\langle S_J, \pi_J \rangle \coloneqq \langle S, \pi \rangle$. By \Cref{th:decomp}\ref{th:decomp:decomp}, the pay gap between the two populations is
\begin{multline*}
W_A(p,q_I,\langle S_I, \pi_I \rangle ) - W_A(p,q_J,\langle S_J, \pi_J \rangle )
\\
\begin{aligned}
&= \bigl[ W_A(p,q_I,\langle S_I, \pi_I \rangle )
- W_A(p,q_J,\langle S_I, \pi_I \rangle ) \bigr]
\\
&\qquad + \mathcal{C}_A(p,q_J,\langle S_J, \pi_J \rangle, \langle S_I, \pi_I \rangle )
\\
&\qquad + \mathcal{I}_A(p,q_J,\langle S_J, \pi_J \rangle, \langle S_I, \pi_I \rangle ) .
\end{aligned}
\end{multline*}
By our previous calculations (see \cpageref{ex:reversal} above), the bracketed term equals $p(1) - p(1) = 0$, the perception-correcting ``$\mathcal{C}_A$'' term equals $p(1) - q_J(1)$, and the instrumental ``$\mathcal{I}_A$'' term is zero, so the total pay gap is $p(1) - q_J(1)$. Thus if population~$J$ is strictly over-perceived ($q_J(1) > p(1)$), then population~$I$ is paid strictly \emph{less} than population~$J$. This holds even if population~$I$ is more favorably perceived ($q_I(1) \geq q_J(1)$).
\end{namedthm}

\subsection{The special case of accurate perceptions}
\label{sec:disc:accurate}

Our focus in this paper is on the interplay of (mis)perceptions and information. In this section, we comment briefly on the special case of accurate perceptions, $q_I = q_J = p$, in which there is no meaningful interaction.

In particular, when perceptions are accurate, the perception-correcting effect is zero for every firm $A \subset \mathcal{A}$:
\begin{multline*}
\mathcal{C}_A(p,p,\langle S_J, \pi_J \rangle ,\langle S_I, \pi_I \rangle )
\\
=
\sum_{s \in S} \mu_p(s)
\sum_{s' \in S'} \left[ \mu_p(s'|s) - \mu_p(s'|s) \right]
\sum_{\theta \in \Theta} \mu_p(\theta|s,s') \widehat{a}_s(\theta)
= 0 .
\end{multline*}
Hence by \Cref{th:decomp}\ref{th:decomp:decomp}--\ref{th:decomp:instr}, 
a more informative population is always paid more when perceptions are accurate. This finding and its converse together constitute Blackwell's theorem on the value of information:%
\footnote{Blackwell calls skill types ``states (of the world),'' tasks ``actions,'' firms ``decision problems,'' and (expected) surplus/pay ``(expected) value/payoff.''}

\begin{theorem}[\cite{Blackwell1951,Blackwell1953}]
\label{th:blackwell}
Consider two populations $I$ and $J$, both with skill distribution~$p$ and with respective perceptions $q_I$ and $q_J$ and signal structures $\langle S_I, \pi_I \rangle $ and $\langle S_J, \pi_J \rangle $. Suppose that perceptions are accurate ($q_I = q_J = p$).
\begin{enumerate}[label=(\alph*)]
\item \label{th:blackwell:proverse}
If $I$ is more informative than $J$ (i.e. $\langle S_I, \pi_I \rangle  \succsim_G \langle S_J, \pi_J \rangle $), then every firm pays population~$I$ more on average:
$$W_A(p,q_I,\langle S_I, \pi_I \rangle )\geq W_A(p,q_J,\langle S_J, \pi_J \rangle )
\quad \text{for every firm $A\subset \mathcal{A}$}.$$
\item \label{th:blackwell:converse}
If $I$ is not more informative than $J$ (i.e. $\langle S_I, \pi_I \rangle  \not\succsim_G \langle S_J, \pi_J \rangle $), then some firm pays population~$I$ strictly less on average:
$$W_A(p,q_I,\langle S_I, \pi_I \rangle )< W_A(p,q_J,\langle S_J, \pi_J \rangle )
\quad \text{for some firm $A\subset \mathcal{A}$}.$$
\end{enumerate}
\end{theorem}

\textcite{ChambersEchenique2021} introduced the model that we employ in this paper, and used it to study statistical discrimination in the accurate-perceptions case. In particular, they sought to characterize the following strong ``no-discrimination'' property: ``every firm $A \subset \mathcal{A}$ pays populations~$I$ and $J$ the same on average, i.e. $W_A(p,p,\langle S_I, \pi_I \rangle) = W_A(p,p,\langle S_J, \pi_J \rangle)$.'' They obtained two geometric characterizations, proved using Choquet theory.

An alternative characterization can be obtained from \hyperref[th:blackwell]{Blackwell's theorem}, which immediately implies that ``no-discrimination'' holds if and only if $\langle S_I, \pi_I \rangle  \succsim_G \langle S_J, \pi_J \rangle \succsim_G \langle S_I, \pi_I \rangle $, meaning that the two populations' signal structures are informationally equivalent (i.e. they convey exactly the same information about skill).

\section{When does information narrow pay gaps?}
\label{sec:gap}

The spirit of the perception-correcting effect is that informative signals diminish the importance for average pay of prior misperceptions. This suggests that the pay gap between two differently-perceived but equally-skilled populations narrows as the informativeness of signals improves. \textcite[][Proposition~1]{BohrenImasRosenberg2019} showed that this intuition is correct in their model of subjective evaluation, which can be thought of as a single-task linear--Gaussian special case of our model. (We describe their model and result in more detail in \Cref{re:bir} below. Their paper contains two further propositions, and its main focus is on empirically testing these propositions.)

In this section, we explore the validity of this ``information narrows pay gaps'' intuition beyond the single-task linear--Gaussian special case, obtaining results that generalize and qualify the finding of \textcite{BohrenImasRosenberg2019}. We focus on the case in which the two populations~$I$ and $J$ have the same signal structure $\langle S, \pi \rangle $, differing only in their perceptions $q_I$ and $q_J$. Our question is under what conditions the pay gap narrows when the populations' common signal structure changes from $\langle S, \pi \rangle $ to a more informative signal structure $\langle S', \pi' \rangle $.

We identify two such conditions. The first, discussed in \cref{sec:gap:slight}, is when the change from $\langle S, \pi \rangle $ to $\langle S', \pi' \rangle $ is small (in particular, small enough not to affect task assignment). The second, discussed in \cref{sec:gap:full}, is when $\langle S', \pi' \rangle $ is close to fully informative. Both conditions are stringent, and we show that when they are not satisfied, extra information may \emph{widen} the pay gap.

These findings have implications for the effectiveness of interventions that seek to reduce labor-market inequality by increasing information availability, such as testing regulations,%
\footnote{For example, Title VII of the 1964 Civil Rights Act and its implementing regulations (Uniform Guidelines on Employee Selection Procedures, 29 C.F.R. § 1607, 1978) require employers to use only tests and selection procedures that have been scientifically validated (or have no disparate impact), incentivizing the use of objective standardized selection measures rather than subjective criteria.}
transparency requirements,%
\footnote{For example, \textcite{CuiLiZhang2020,LaouenanRathelot2022} evaluate how providing extra information impacts discrimination on Airbnb. For lab experiments on such information interventions, see the survey by \textcite{LitwinLow2025}.}
and the introduction of algorithmic screening tools.%
\footnote{For example, \textcite{CowgillTucker2020} discuss how algorithmic screening may alleviate discrimination.}
Our findings suggest that such interventions may either succeed or backfire: it depends on perceptions and on how extra information reshapes task assignment.

\subsection{Slightly increased informativeness}
\label{sec:gap:slight}

We call a change in informativeness \emph{slight} if it is small enough that it does not change the surplus-maximizing task assignment. The formal definition is as follows.

\begin{definition}
\label{def:slight}
Fix a firm $A \subset \mathcal{A}$, a perception $q$, and signal structures $\langle S', \pi' \rangle \succsim_G \langle S, \pi \rangle $, and let $g$ be a garbling kernel from $\langle S',\pi' \rangle $ to $\langle S,\pi \rangle $. We say that $\langle S',\pi' \rangle $ is \emph{slightly more informative than} $\langle S,\pi \rangle $ for firm~$A$ at perception $q$ if $\widehat{a}_s = \widehat{a}'_{s'}$ for all signals $s \in S$ and $s' \in S'$ such that $g(s|s')>0$.
\end{definition}

By definition, slight increases of informativeness leave task assignment unchanged, so that the instrumental component $\mathcal{I}_A(p,q,\langle S, \pi \rangle ,\langle S', \pi' \rangle )$ of the change of average pay is zero. Note that for single-task firms $A = \{a\}$, \emph{all} increases of informativeness are slight (whatever the perception), since such firms' task assignment is immutable.

\begin{remark}
\label{re:local_suff}
Fix a firm~$A \subset \mathcal{A}$ (a finite set, by definition) and a perception $q$. Pay as a function of the posterior belief, $r \mapsto \max_{a \in A} \sum_{\theta \in \Theta} r(\theta) a(\theta)$, is locally affine at almost every belief $r' \in \Delta(\Theta)$.%
\footnote{By ``locally affine at $r'$,'' we mean that there exists a neighborhood of $r'$ on which pay is affine. By ``almost every,'' we mean according to the Lebesgue measure on $\Delta(\Theta)$.}
Hence generic signal structures $\langle S, \pi \rangle $ exclusively generate posterior beliefs at which pay is locally affine.%
\footnote{Formally, for any non-empty finite set $S$, it holds for (Lebesgue-)almost every $\pi : \Theta \times S \to \mathbb{R}$ such that $\langle S, \pi \rangle $ is a signal structure that for each signal $s \in S$, pay is locally affine at the posterior belief $q_{\langle S, \pi \rangle }(\cdot|s) \in \Delta(\Theta)$.}
For such generic signal structures $\langle S, \pi \rangle $, any more-informative signal structure $\langle S',\pi' \rangle $ such that the posterior beliefs $q_{\langle S, \pi \rangle }(\cdot|s)$ and $q_{\langle S', \pi' \rangle }(\cdot|s')$ are always sufficiently close is slightly more informative than $\langle S, \pi \rangle $ at perception~$q$.%
\footnote{Explicitly: there exists an $\varepsilon > 0$ such that if $\langle S',\pi' \rangle \succsim_G \langle S,\pi \rangle $ and $\lvert q_{\langle S', \pi' \rangle }(\theta|s') - q_{\langle S, \pi \rangle }(\theta|s) \rvert \leq \varepsilon$ for all $\theta \in \Theta$ and all pairs $(s,s') \in S \times S'$ with $\mu_q(s,s') > 0$, then $\langle S',\pi' \rangle $ is slightly more informative than $\langle S,\pi \rangle $ for firm~$A$ at perception~$q$.}
\end{remark}

The following ``positive'' result identifies (admittedly stringent) conditions under which extra information does indeed narrow the pay gap.

\begin{proposition}
\label{prop:slight}
Fix a firm $A \subset \mathcal{A}$ and two populations $I$ and $J$, both with skill distribution~$p$, signal structure $\langle S_I,\pi_I\rangle = \langle S_J,\pi_J\rangle = \langle S,\pi\rangle$, and respective perceptions $q_I$ and $q_J$. Suppose that population~$I$ is more favorably perceived ($q_I \succsim_{LR} q_J$), and fix a more-informative signal structure $\langle S', \pi' \rangle$ ($\langle S', \pi' \rangle \succsim_G \langle S, \pi \rangle $). Assume that
\begin{enumerate}[label=(\roman*)]
\item \label{prop:slight:mon} the firm $A$ is monotone ($A \subset \mathcal{A}_M$),
\item \label{prop:slight:mlr} the signal structure $\langle S', \pi' \rangle $ is MLR,
\item \label{prop:slight:over} population~$I$ is over-perceived ($q_I \succsim_{LR} p$),
\item \label{prop:slight:under} population~$J$ is under-perceived ($p \succsim_{LR} q_J$), and
\item \label{prop:slight:slight} $\langle S', \pi' \rangle $ is slightly more informative than $\langle S, \pi \rangle $ at both $q_I$ and $q_J$.
\end{enumerate}
Then the extra information narrows the pay gap:
\begin{align}
\nonumber
& W_A(p,q_I,\langle S',\pi' \rangle )-W_A(p,q_J,\langle S',\pi' \rangle )
\\
\tag{$\star$}\label{eq:narrow}
\leq{}& W_A(p,q_I,\langle S,\pi \rangle )-W_A(p,q_J,\langle S,\pi \rangle ) .
\end{align}
\end{proposition}

None of the assumptions in \Cref{prop:slight} can be dispensed with: if any one of them fails, then extra information may \emph{widen} the pay gap, as the following result shows.

\begin{proposition}
\label{prop:tight}
Consider tuples $( A, p, q_I, q_J, \langle S, \pi \rangle , \langle S', \pi' \rangle )$ comprising a firm $A \subset \mathcal{A}$, a skill distribution $p$, perceptions $q_I$ and $q_J$ satisfying $q_I \succsim_{LR} q_J$, and signal structures $\langle S, \pi \rangle $ and $\langle S', \pi' \rangle $ satisfying $\langle S', \pi' \rangle \succsim_G \langle S, \pi \rangle $. For each of the properties~\ref{prop:slight:mon}--\ref{prop:slight:slight} in \Cref{prop:slight}, there exist tuples which violate that property, satisfy all four of the other properties, and violate~\eqref{eq:narrow}.
\end{proposition}

The role of property~\ref{prop:slight:slight} is to ensure that the instrumental effect of extra information is zero for both populations, so that the change in the pay gap equals the difference between the two populations' perception-correcting effects. The pay gap then narrows if the perception-correcting effect is negative for population~$I$ and positive for population~$J$. The role of properties~\ref{prop:slight:over} and \ref{prop:slight:under} is to ensure that this is the case. Formally, signing the perception-correcting terms in this way also requires properties~\ref{prop:slight:mon} and \ref{prop:slight:mlr}, as explained in section \ref{sec:decomp:mlr_mon} above.

To see why property~\ref{prop:slight:slight} is necessary, consider the following tuple. Skill types are binary, $\Theta = \{0,1\}$. The skill distribution and perceptions are $p(1)=1/2$, $q_I(1)=3/4$ and $q_J(1)=1/4$, so that population~$I$ is over-perceived and population~$J$ is under-perceived. We consider the monotone firm $A = \{\overline{a},\widetilde{a}\}$, where $\overline{a}(\theta) \coloneqq \theta$ and $\widetilde{a}(\theta) \coloneqq 4(2\theta-1)$ for each $\theta \in \Theta$. For each $\lambda \in [1/2,1]$, let $\langle S_\lambda, \pi_\lambda \rangle $ be the MLR signal structure given by $S_\lambda = \bigl\{s^0,s^1\bigr\}$ and $\pi_\lambda\bigl(s^0\bigm|0\bigr)=\pi_\lambda\bigl(s^1\bigm|1\bigr)=\lambda$. Note that a higher value of $\lambda$ corresponds to greater informativeness. Each population's average pay is plotted in \Cref{fig:tight-slight} as a function of $\lambda$.%
\footnote{The kinks in \Cref{fig:tight-slight} correspond to changes in task assignment. In particular, in population~$I$, workers with signal $s^1$ are assigned to task $\widetilde{a}$ whatever the value of $\lambda$, while those with signal $s^0$ are assigned to $\widetilde{a}$ if $\lambda < 9/13$ and to $\overline{a}$ if $\lambda > 9/13$, and in population~$J$, workers with signal $s^0$ are assigned to task $\overline{a}$ whatever the value of $\lambda$, while those with signal $s^1$ are assigned to $\overline{a}$ if $\lambda < 4/5$ and to $\widetilde{a}$ if $\lambda > 4/5$.}
The pay gap is not decreasing: there are $\lambda'>\lambda$ in $[1/2,1]$ such that the pay gap is greater at $\lambda'$ than at $\lambda$.

\begin{figure}
\centering
\begin{tikzpicture}[scale=1, line cap=round,
declare function={post_lo(\x,\q)= \q * \x / ( (1-\q) * (1-\x) + \q * \x ); post_hi(\x,\q)= \q * (1-\x) / ( (1-\q) * \x + \q * (1-\x) ); pay(\x,\q,\K) = 0.5 * max( post_lo(\x,\q), \K*(2*post_lo(\x,\q)-1) ) + 0.5 * max( post_hi(\x,\q), \K*(2*post_hi(\x,\q)-1) );}]
    \pgfmathsetmacro{\xscale}{15};
    \pgfmathsetmacro{\K}{4};
    \pgfmathsetmacro{\yscale}{8/\K};
    \pgfmathsetmacro{\kinkI}{ (3/4)*(1-(\K/(2*\K-1))) / ( (3/4) + (1-2*(3/4))*(\K/(2*\K-1)) ) }; 
    \pgfmathsetmacro{\kinkJ}{ ( (1-(1/4))*(\K/(2*\K-1) ) / ( (1/4) + (1-2*(1/4))*(\K/(2*\K-1) ) }; 
    \draw[domain=0:0.5, variable=\x, samples=100, very thick] plot ( {\xscale*(1-\x)}, { \yscale*pay(\x,0.75,\K) } );
    \draw ( {\xscale*(1-0.4)}, {\yscale*(pay(0.4,0.75,\K)+0.05)} ) node[anchor=south] {population $I$};
    \draw[domain=0:0.5, variable=\x, samples=100, very thick,] plot ( {\xscale*(1-\x)}, { \yscale*pay(\x,0.25,\K) } );
    \draw ( {\xscale*(1-0.4)}, { \yscale*pay(0.4,0.25,\K) } ) node[anchor=south] {population $J$};
    \draw[domain=0:0.5, variable=\x, samples=100, very thick, dashed] plot ( {\xscale*(1-\x)}, { \yscale*(pay(\x,0.75,\K)-pay(\x,0.25,\K)) } );
    \draw ( {\xscale*(1-0.03)}, {\yscale*(pay(0.03,0.75,\K)-pay(0.03,0.25,\K))} ) node[anchor=south west] {gap};
    \draw[<-] ({\xscale*0.5-0.8},{\yscale*pay(0.4,0.75,\K)+0.5}) -- ({\xscale*0.5-0.8},{1*0})-- ({\xscale*0.5-0.55},{1*0});
    \draw[-] ({\xscale*0.5-0.45},{1*(-0.15)}) -- ({\xscale*0.5-0.35},{1*0.15});
    \draw[-] ({\xscale*0.5-0.60},{1*(-0.15)}) -- ({\xscale*0.5-0.50},{1*0.15});
    \draw ({\xscale*1+0.5},{1*0}) node[anchor=west] {$\lambda$};
    \draw[->] ({\xscale*0.5-0.40},{1*0}) -- ({\xscale*0.5-0.4},{1*0}) -- ({\xscale*1+0.5},{1*0});
    \draw[-] ({\xscale*0.5},{1*0.07}) -- ({\xscale*0.5},{1*(-0.07)});
    \draw ({\xscale*1+0.5},{1*0}) node[anchor=west] {$\lambda$};
    \draw[-] ({\xscale*0.5},{1*0.07}) -- ({\xscale*0.5},{1*(-0.07)});
    \draw ({\xscale*0.5},{1*0}) node[anchor=north] {$\strut 1/2$};
    \draw[-] ({\xscale*\kinkI},{1*0.07}) -- ({\xscale*\kinkI},{1*(-0.07)});
    \draw ({\xscale*\kinkI},{1*0}) node[anchor=north] {$\strut 9/13$};
    \draw[-] ({\xscale*\kinkJ},{1*0.07}) -- ({\xscale*\kinkJ},{1*(-0.07)});
    \draw ({\xscale*\kinkJ},{1*0}) node[anchor=north] {$\strut 4/5$};    
    \draw ({\xscale*1},{1*0}) node[anchor=north] {$\strut 1$};
\end{tikzpicture}
\caption{Average pay of populations~$I$ and $J$ and the gap between them in the tuple used to prove the indispensability of property~\ref{prop:slight:slight}.}
\label{fig:tight-slight}
\end{figure}

Taken together, \Cref{prop:slight,prop:tight} show that whether additional information will narrow the pay gap very much depends.

\begin{remark}
\label{re:bir}
In this section, where we consider only slight increases of informativeness (which leave task assignment unchanged), little insight is lost by focusing on single-task firms $A=\{a\}$. This special case is close to the model of \textcite{BohrenImasRosenberg2019}, which aims to describe subjective evaluation rather than pay, and accordingly features no production choices such as task assignment. In these authors' model, it is additionally assumed that the primitives $a$, $p$, $q$ and $\langle S, \pi \rangle $ are linear--Gaussian.%
\footnote{To be precise: $\Theta=S=\mathbb{R}$, $a(\theta)=\theta$ for every $\theta \in \Theta$, and $(\theta,s) \mapsto p(\theta) \pi(s|\theta)$ and $(\theta,s) \mapsto q(\theta) \pi(s|\theta)$ are bi-variate Gaussian probability density functions.}
In our language, Bohren, Imas and Rosenberg's first result implies that in the single-task linear--Gaussian case, extra information narrows the pay gap. Our \Cref{prop:slight,prop:tight} generalize and qualify this finding; for example, without the specific linear--Gaussian structure, extra information may \emph{widen} the pay gap unless population~$I$ is over-perceived and population~$J$ is under-perceived.
\end{remark}

\begin{namedthm}[\Cref*{re:unit} {\normalfont(continued from \cpageref{re:unit})}.]
\label{re:unit_gap}
Recall our alternative interpretation of the model, in which $q_I$ and $q_J$ are the two populations' true skill distributions, and a ``like-for-like'' comparison is being made between two sub-populations whose skill distributions are equal, $p_I=p_J=p$. One upshot of \Cref{prop:slight,prop:tight} is that extra information may narrow the pay gap in some sub-populations (those with $q_I \succsim_{LR} p \succsim_{LR} q_J$) and widen it in others (those with $p \succsim_{LR} q_I$ or $q_J \succsim_{LR} p$).

In correspondence studies such as \textcite{BertrandMullainathan2004}, in which the researcher chooses a distribution of CVs to be sent to firms and compares outcomes across two populations, \Cref{prop:slight,prop:tight} imply that whether extra information narrows or widens the measured outcome gap may depend on exactly which CV distribution was chosen by the researcher.
\end{namedthm}

\subsection{Near-full informativeness}
\label{sec:gap:full}

Our second sufficient condition for extra information to narrow the pay gap is for the new, more informative signal structure to be ``nearly fully informative,'' formally defined as follows.

\begin{definition}
\label{def:nearlyfull}
Given $\varepsilon \geq 0$, we say that a signal structure $\langle S,\pi\rangle$ is \emph{within $\varepsilon$ of full information} if for each signal $s \in S$, there exists a skill type $\theta_s \in \Theta$ such that $\pi(s|\theta) \leq \varepsilon \pi(s|\theta_s)$ for every other skill type $\theta \in \Theta \setminus \{\theta_s\}$.%
\footnote{Equivalently: $\min_{\substack{\theta \in \Theta : \pi(s|\theta)>0}} [ \max_{\theta' \in \Theta \setminus \{\theta\}} \pi(s|\theta') / \pi(s|\theta) ] \leq \varepsilon$ for every $s \in S$.}
\end{definition}

The following ``positive'' result shows that when informativeness is increased to near-full, the pay gap narrows.

\begin{proposition}
\label{prop:nearlyfull}
Fix a firm $A \subset \mathcal{A}$ and two populations $I$ and $J$, both with skill distribution~$p$, signal structure $\langle S_I,\pi_I\rangle = \langle S_J,\pi_J\rangle = \langle S,\pi\rangle$, and respective perceptions $q_I$ and $q_J$. Suppose that population~$I$ is more favorably perceived ($q_I \succsim_{LR} q_J$). There exists an $\varepsilon \geq 0$ such that any signal structure $\langle S',\pi' \rangle $ within $\varepsilon$ of full information satisfies \eqref{eq:narrow} on \cpageref{eq:narrow}, that is, yields a narrower pay gap than $\langle S,\pi \rangle $. Furthermore, $\varepsilon$ may be chosen to be strictly positive, except if the right-hand side of \eqref{eq:narrow} is equal to zero.
\end{proposition}

\Cref{prop:nearlyfull} is a continuity result. Under a fully informative signal structure $\langle S'', \pi'' \rangle$, each signal $s'' \in S''$ perfectly reveals some skill type $\theta_{s''} \in \Theta$ (that is, $q_{\langle S'', \pi'' \rangle }(\theta_{s''}|s'')=1$), so workers are paid
$$w_A(s'',q_I,\langle S'', \pi'' \rangle ) = \max_{a \in A} a(\theta_{s''}) = w_A(s'',q_J,\langle S'', \pi'' \rangle ) .$$
Hence there is no pay gap: $W_A(p,q_I,\langle S'',\pi'' \rangle ) - W_A(p,q_J,\langle S'',\pi'' \rangle ) = 0$. Since pay as a function of the posterior belief, $r \mapsto \max_{a \in A} \sum_{\theta \in \Theta} r(\theta) a(\theta)$, is continuous, it follows that the pay gap $W_A(p,q_I,\langle S',\pi' \rangle ) - W_A(p,q_J,\langle S',\pi' \rangle )$ can be made arbitrarily small (in particular, small enough that \eqref{eq:narrow} holds) by choosing a signal structure $\langle S', \pi' \rangle $ which produces sufficiently extreme beliefs. And the beliefs produced by signal structures $\langle S', \pi' \rangle $ that are within $\varepsilon$ of full information (in fact) become arbitrarily extreme as $\varepsilon>0$ shrinks.

\section{Concluding thoughts}
\label{sec:concl}

We have studied the interplay of information and prior misperceptions in shaping average pay. Our main theoretical contribution was a decomposition of the effect of additional information into a familiar instrumental term à la \textcite{Blackwell1951,Blackwell1953} and a perception-correcting term arising from misperception. This decomposition has implications for statistical discrimination.

In order to isolate the interaction between information and misperception, we have deliberately abstracted away from other realistic frictions, such as misperceptions about signal structures \parencite[as in][]{BohrenHaggagImasPope2025}, monopsony power in the labor market,%
\footnote{See e.g. \textcite{BoalRansom1997,Manning2003,Manning2021,Card2022,BergerHerkenhoffMongey2022,YehMacalusoHershbein2022}.}
and agency frictions inside the firm.%
\footnote{See e.g. \textcite{GibbonsRoberts2013,Mookherjee2006,Holmstrom2017}.}
Extending our decomposition to encompass these or other frictions is in principle straightforward: for each additional friction, there is an extra term in the decomposition capturing how new information impacts that friction, holding fixed both task assignment and the operation of the other frictions.



\begin{appendices}
\crefalias{section}{appsec}
\crefalias{subsection}{appsec}
\crefalias{subsubsection}{appsec}

\section{Proof of \texorpdfstring{\Cref{le:favor}}{Lemma \ref{le:favor}}}
\label{app:le:favor_pf}

For part~\ref{le:favor:suff}, suppose that $q' \succsim_{LR} q$. It suffices to show that
$$w_A(s,q',\langle S, \pi \rangle ) \geq w_A(s,q,\langle S, \pi \rangle )$$
holds for any signal structure $\langle S, \pi \rangle $, any signal $s \in S$, and any monotone firm $A \subset \mathcal{A}_M$. To that end, fix a signal structure $\langle S, \pi \rangle $ and a signal $s \in S$. Since $q' \succsim_{LR} q$, we have for any skill types $\theta'' \geq \theta'$ in $\Theta$ that
\begin{align*}
q_{\langle S, \pi \rangle }(\theta'|s)
q'_{\langle S, \pi \rangle }(\theta''|s)
&= \pi(s|\theta') \pi(s|\theta'')
q(\theta') q'(\theta'') / k
\\
&\geq \pi(s|\theta') \pi(s|\theta'')
q(\theta'') q'(\theta') / k
= q_{\langle S, \pi \rangle }(\theta''|s)
q'_{\langle S, \pi \rangle }(\theta'|s) 
\end{align*}
where $k>0$ is a constant, which is to say that $q'_{\langle S, \pi \rangle }(\cdot|s) \succsim_{LR} q_{\langle S, \pi \rangle }(\cdot|s)$. Hence $q'_{\langle S, \pi \rangle }(\cdot|s)$ first-order stochastically dominates $q_{\langle S, \pi \rangle }(\cdot|s)$, so for any monotone firm $A \subset \mathcal{A}_M$, letting
$\widehat{a}_s
\in \argmax_{a \in A} \sum_{\theta \in \Theta} q_{\langle S, \pi \rangle }(\theta|s) a(\theta)$,
we have
\begin{multline*}
w_A(s,q',\langle S, \pi \rangle )
= \max_{a \in A} \sum_{\theta \in \Theta} q'_{\langle S, \pi \rangle }(\theta|s) a(\theta)
\geq \sum_{\theta \in \Theta} q'_{\langle S, \pi \rangle }(\theta|s) \widehat{a}_s(\theta)
\\
\geq \sum_{\theta \in \Theta} q_{\langle S, \pi \rangle }(\theta|s) \widehat{a}_s(\theta)
= w_A(s,q,\langle S, \pi \rangle ) ,
\end{multline*}
where the second inequality holds since $\widehat{a}_s \in \mathcal{A}_M$ and $q'_{\langle S, \pi \rangle }(\cdot|s)$ first-order stochastically dominates $q_{\langle S, \pi \rangle }(\cdot|s)$.

For part~\ref{le:favor:nec}, suppose that $q' \not\succsim_{LR} q$. Then there exist $\theta'' > \theta'$ in $\Theta$ such that $q'(\theta'')/q'(\theta') < q(\theta'')/q(\theta')$. Let $\langle S, \pi \rangle $ be the signal structure that sends message $s$ if skill is $\theta \in \{\theta',\theta''\}$, and otherwise perfectly reveals skill: $S = \{s\} \cup \{s_\theta\}_{\theta \in \Theta \setminus \{\theta',\theta''\}}$, $\pi(s|\theta')=\pi(s|\theta'')=1$ and $\pi(s_\theta|\theta)=1$ for every $\theta \in \Theta \setminus \{\theta',\theta''\}$. Fix any monotone firm $A \subset \mathcal{A}_M$. Following the imperfectly revealing signal $s \in S$, we have $q'_{\langle S, \pi \rangle }(\cdot|s) \prec_{LR} q_{\langle S, \pi \rangle }(\cdot|s)$ since
$$\frac{q'_{\langle S, \pi \rangle }(\theta''|s)}{q'_{\langle S, \pi \rangle }(\theta'|s)}
= \frac{\pi(s|\theta'')}{\pi(s|\theta')}
\frac{q'(\theta'')}{q'(\theta')}
< \frac{\pi(s|\theta'')}{\pi(s|\theta')}
\frac{q(\theta'')}{q(\theta')} 
= \frac{q_{\langle S, \pi \rangle }(\theta''|s)}{q_{\langle S, \pi \rangle }(\theta'|s)}.$$
Hence $q'_{\langle S, \pi \rangle }(\cdot|s)$ is strictly first-order stochastically dominated by $q_{\langle S, \pi \rangle }(\cdot|s)$, so letting
$\widehat{a}'_s
\in \argmax_{a \in A} \sum_{\theta \in \Theta} q'_{\langle S, \pi \rangle }(\theta|s) a(\theta)$,
we have
\begin{multline*}
w_A(s,q',\langle S, \pi \rangle )
= \sum_{\theta \in \Theta} q'_{\langle S, \pi \rangle }(\theta|s) \widehat{a}'_s(\theta)
< \sum_{\theta \in \Theta} q_{\langle S, \pi \rangle }(\theta|s) \widehat{a}'_s(\theta)
\\
\leq \max_{a \in A} \sum_{\theta \in \Theta} q_{\langle S, \pi \rangle }(\theta|s) a(\theta)
= w_A(s,q,\langle S, \pi \rangle ),
\end{multline*}
where the strict inequality holds since $\widehat{a}'_s \in \mathcal{A}_M$ and $q'_{\langle S, \pi \rangle }(\cdot|s)$ is strictly first-order stochastically dominated by $q_{\langle S, \pi \rangle }(\cdot|s)$. When skill is fully revealed, average pay does not depend on the prior perception: for each $\theta \in \Theta \setminus \{\theta',\theta''\}$, $q'_{\langle S, \pi \rangle }(\theta|s_\theta) = q_{\langle S, \pi \rangle }(\theta|s_\theta) = 1$, so
$$w_A(s_{\theta},q',\langle S, \pi \rangle )
= \max_{a \in A} a(\theta)
= w_A(s_{\theta},q,\langle S, \pi \rangle ) .$$
Hence $W_A(p,q',\langle S, \pi \rangle ) < W_A(p,q,\langle S, \pi \rangle )$.
\qed

\section{Proof of \texorpdfstring{\Cref{th:decomp}}{Theorem \ref{th:decomp}}}
\label{app:th:decomp_pf}

For part~\ref{th:decomp:decomp}, define
\begin{equation*}
K \coloneqq \sum_{s\in S} \mu_p(s) \sum_{s'\in S'} \mu_p(s'|s)
\sum_{\theta\in\Theta} \mu_q(\theta|s,s') \widehat{a}_{s}(\theta) .
\end{equation*}
Firstly,
\begin{align*}
W_A(p,q,\langle S, \pi \rangle )
&=\sum_{s\in S} \mu_p(s) w_A(s,q,\langle S, \pi \rangle )
\\
&=\sum_{s\in S} \mu_p(s)
\sum_{\theta \in \Theta} \left[ \sum_{s' \in S'} \mu_q(s'|s) \mu_q(\theta|s,s') \right] \widehat{a}_s(\theta)
\\
&= \sum_{s\in S} \mu_p(s) \sum_{s'\in S'} \mu_q(s'|s)
\sum_{\theta\in\Theta} \mu_q(\theta|s,s') \widehat{a}_{s}(\theta) 
\\
&= K - \mathcal{C}_A(p,q,\langle S, \pi \rangle ,\langle S', \pi' \rangle ),
\end{align*}
where the second equality holds because $q_{\langle S, \pi \rangle }(\theta|s) = \mu_q(\theta|s)$, which equals the bracketed term by the law of total probability. Secondly,
\begin{align*}
W_A(p,q,\langle S', \pi' \rangle )
&= \sum_{s'\in S'} \mu_p(s') w_A(s',q,\langle S', \pi' \rangle )
\\
&= \sum_{s\in S} \mu_p(s)
\sum_{s'\in S'} \mu_p(s'|s)
\sum_{\theta\in\Theta} q_{\langle S', \pi' \rangle }(\theta|s') \widehat{a}'_{s'}(\theta)
\\
&= \sum_{s\in S} \mu_p(s)
\sum_{s'\in S'} \mu_p(s'|s)
\sum_{\theta\in\Theta} \mu_q(\theta|s,s') \widehat{a}'_{s'}(\theta)
\\
&= K + \mathcal{I}_A(p,q,\langle S, \pi \rangle ,\langle S', \pi' \rangle ) ,
\end{align*}
where the third equality holds since $q_{\langle S', \pi' \rangle }(\theta|s') = \mu_q(\theta|s') = \mu_q(\theta|s,s')$.

For part~\ref{th:decomp:instr}, we have for all $s \in S$ and $s' \in S'$ that
\begin{multline*}
\sum_{\theta \in \Theta} \mu_q(\theta|s,s') \widehat{a}'_{s'}(\theta)
= \sum_{\theta \in \Theta} q_{\langle S', \pi' \rangle }(\theta|s') \widehat{a}'_{s'}(\theta)
= \max_{a \in A} \sum_{\theta \in \Theta} q_{\langle S', \pi' \rangle }(\theta|s') a(\theta)
\\
\geq \sum_{\theta \in \Theta} q_{\langle S', \pi' \rangle }(\theta|s') \widehat{a}_s(\theta)
= \sum_{\theta \in \Theta} \mu_q(\theta|s,s') \widehat{a}_s(\theta) ,
\end{multline*}
and thus
\begin{multline*}
\mathcal{I}_A(p,q,\langle S, \pi \rangle ,\langle S', \pi' \rangle )
\\
= \sum_{s \in S} \sum_{s' \in S'} \mu_p(s,s')
\sum_{\theta \in \Theta} \mu_q(\theta|s,s') \left[ \widehat{a}'_{s'}(\theta) - \widehat{a}_s(\theta) \right]
\geq 0 .
\end{multline*}

For part~\ref{th:decomp:corr_pos}, suppose in addition that $A \subset \mathcal{A}_M$, that $\langle S', \pi' \rangle $ is MLR, and that $p \succsim_{LR} q$. It is enough to show that
\begin{equation*}
\sum_{s'\in S'} [ \mu_p(s'|s) - \mu_q(s'|s)  ] \sum_{\theta\in\Theta} \mu_q(\theta|s,s') \widehat{a}_{s}(\theta)
\geq 0 \quad \text{for every $s \in S$.}
\end{equation*}
To that end, we fix an arbitrary $s \in S$ and establish that
\begin{enumerate}
\item \label{th:decomp:pf_incr} $s' \mapsto \sum_{\theta\in\Theta} \mu_q(\theta|s,s') \widehat{a}_{s}(\theta)$
is increasing, and that
\item \label{th:decomp:pf_fosd} $s' \mapsto \mu_p(s'|s)$ first-order stochastically dominates $s' \mapsto \mu_q(s'|s)$.
\end{enumerate}

To establish claim~\ref{th:decomp:pf_incr}, note that for any $t'>s'$ in $S'$, $\theta \mapsto \mu_q(\theta|s,t')$ is more favorable than $\theta \mapsto \mu_q(\theta|s,s')$ since for any $\theta''>\theta'$ in $\Theta$,
\begin{align*}
\mu_q(\theta'|s,s')
\mu_q(\theta''|s,t')
&= q_{\langle S', \pi' \rangle }(\theta'|s')
q_{\langle S', \pi' \rangle }(\theta''|t')
\\
&= \pi'(s'|\theta') \pi'(t'|\theta'') q(\theta') q(\theta'') / k
\\
&\geq \pi'(t'|\theta') \pi'(s'|\theta'') q(\theta') q(\theta'') / k
\\
&= q_{\langle S', \pi' \rangle }(\theta'|t')
q_{\langle S', \pi' \rangle }(\theta''|s')
= \mu_q(\theta'|s,t')
\mu_q(\theta''|s,s') 
\end{align*}
for a constant $k>0$, where the first and last equalities use the fact that $\theta$ is independent of $s$ conditional on $s'$ under $\mu_q$, and the inequality holds since $\langle S', \pi' \rangle $ is MLR. Hence $\theta \mapsto \mu_q(\theta|s,t')$ first-order stochastically dominates $\theta \mapsto \mu_q(\theta|s,s')$. Because $\widehat{a}_s \in \mathcal{A}_M$, claim~\ref{th:decomp:pf_incr} follows.

To establish claim~\ref{th:decomp:pf_fosd}, define
\begin{equation*}
\kappa(s'|\theta) \coloneqq \frac{ \pi'(s'|\theta) g(s|s') }{ \sum_{t' \in S'} \pi'(t'|\theta) g(s|t') } 
\quad \text{for each $s' \in S'$ and $\theta \in \Theta$,}
\end{equation*}
and use bars to denote CDFs, in particular
\begin{equation*}
\overline{\pi'}(s'|\theta) \coloneqq \sum_{t' \in S' : t' \leq s'} \pi'(t'|\theta) \quad \text{and} \quad
\overline{\kappa}(s'|\theta) \coloneqq \sum_{t' \in S' : t' \leq s'} \kappa(t'|\theta)
\end{equation*}
for all $s' \in S'$ and $\theta \in \Theta$. For each $s' \in S'$, we have
\begin{equation*}
\mu_p(s'|s)
= \sum_{\theta \in \Theta} \mu_p(\theta|s) \mu_p(s'|\theta,s)
= \sum_{\theta \in \Theta} p_{\langle S, \pi \rangle }(\theta|s) \kappa(s'|\theta) .
\end{equation*}
Performing the same calculation for $\mu_q$ and subtracting yields
$$\mu_p(s'|s) - \mu_q(s'|s)
= \sum_{\theta \in \Theta} \left[ p_{\langle S, \pi \rangle }(\theta|s) - q_{\langle S, \pi \rangle }(\theta|s) \right] \kappa(s'|\theta) \quad \text{for each $s' \in S'$,}$$
and thus
\begin{equation*}
\overline{\mu_p}(s'|s) - \overline{\mu_q}(s'|s)
= \sum_{\theta \in \Theta} \left[ p_{\langle S, \pi \rangle }(\theta|s) - q_{\langle S, \pi \rangle }(\theta|s) \right] \overline{\kappa}(s'|\theta) \quad \text{for each $s' \in S'$.}
\end{equation*}
As shown in the proof of \Cref{le:favor}, the fact that $p \succsim_{LR} q$ implies that $p_{\langle S, \pi \rangle }(\cdot|s)$ is more favorable than $q_{\langle S, \pi \rangle }(\cdot|s)$, which in turn implies that $p_{\langle S, \pi \rangle }(\cdot|s)$ first-order stochastically dominates $q_{\langle S, \pi \rangle }(\cdot|s)$. Hence to establish claim~\ref{th:decomp:pf_fosd}, it suffices to show that for each $s' \in S'$, $\theta \mapsto \overline{\kappa}(s'|\theta)$ is decreasing.

To that end, fix $\theta' > \theta$ in $\Theta$, and enumerate the elements of $S'$ as $S' = \{s_1',s_2',\dots,s_n'\}$, where $s_1' < s_2' < \cdots < s_n'$; we will show that $\overline{\kappa}(s_i'|\theta) \geq \overline{\kappa}(s_i'|\theta')$ for every $i \in \{1,2,\dots,n\}$. For each $i \in \{1,2,\dots,n\}$, define $b_i \coloneqq \pi'(s_i'|\theta) g(s|s_i')$ and $b_i' \coloneqq \pi'(s_i'|\theta') g(s|s_i')$. Because $\langle S', \pi' \rangle $ is MLR, we have $b_j b_k' \geq b_j' b_k$ whenever $j < k$. Hence for each $i \in \{1,2,\dots,n\}$, it holds that
$$\sum_{j \leq i} \sum_{k>i} \left( b_j b_k' - b_j' b_k \right) \geq 0 ,$$
or equivalently
$$\left( \sum_{j \leq i} b_j \right) \left( \sum_k b_k' \right) \geq \left( \sum_{j \leq i} b_j' \right) \left( \sum_k b_k \right) ,$$
whence
$$\overline{\kappa}(s_i'|\theta) = \frac{ \sum_{j \leq i} b_j }{ \sum_k b_k } \geq \frac{ \sum_{j \leq i} b_j' }{ \sum_k b_k' } = \overline{\kappa}(s_i'|\theta') .$$

Finally, for part~\ref{th:decomp:corr_neg}, suppose that $q \succsim_{LR} p$. Analogously to part~\ref{th:decomp:corr_pos}, it suffices to fix an arbitrary $s \in S$ and to establish claim~\ref{th:decomp:pf_incr} above and
\begin{enumerate}[label=(\arabic*$'$)]
\setcounter{enumi}{1}
\item \label{th:decomp:pf_fosd_rev} that $s' \mapsto \mu_q(s'|s)$ first-order stochastically dominates $s' \mapsto \mu_p(s'|s)$.
\end{enumerate}
Claim~\ref{th:decomp:pf_incr} follows from exactly the argument above, while claim~\ref{th:decomp:pf_fosd_rev} follows from applying the above argument for claim~\ref{th:decomp:pf_fosd} except with $p \succsim_{LR} q$ replaced by $q \succsim_{LR} p$. \qed

\section{Proof of \texorpdfstring{\Cref{prop:slight}}{Proposition \ref{prop:slight}}}
\label{app:prop:slight_pf}

By definition of ``slightly more informative than,'' assumption~\ref{prop:slight:slight} implies that each population's instrumental component is zero:
$$\mathcal{I}_A(p,q_I,\langle S, \pi \rangle ,\langle S', \pi' \rangle )
= \mathcal{I}_A(p,q_J,\langle S, \pi \rangle ,\langle S', \pi' \rangle )
= 0 .$$
Hence
\begin{align*}
&\bigl[ W_A(p,q_I,\langle S',\pi' \rangle )-W_A(p,q_J,\langle S',\pi' \rangle ) \bigr]
\\
&\quad - \bigl[ W_A(p,q_I,\langle S,\pi \rangle )-W_A(p,q_J,\langle S,\pi \rangle ) \bigr]
\\
&\qquad\qquad = \mathcal{C}_A(p,q_I,\langle S, \pi \rangle ,\langle S', \pi' \rangle )
- \mathcal{C}_A(p,q_J,\langle S, \pi \rangle ,\langle S', \pi' \rangle )
\leq 0 ,
\end{align*}
where the equality holds by \Cref{th:decomp}\ref{th:decomp:decomp}, and the inequality holds because the first ``$\mathcal{C}_A$'' term is non-positive by \Cref{th:decomp}\ref{th:decomp:corr_neg} (applicable by assumptions~\ref{prop:slight:mon}--\ref{prop:slight:over}) and the second ``$\mathcal{C}_A$'' term is non-negative by \Cref{th:decomp}\ref{th:decomp:corr_pos} (applicable by assumptions~\ref{prop:slight:mon}--\ref{prop:slight:mlr} and \ref{prop:slight:under}).
\qed

\section{Proof of \texorpdfstring{\Cref{prop:tight}}{Proposition \ref{prop:tight}}}
\label{app:prop:tight_pf}

For property~\ref{prop:slight:mon}, return to \Cref{ex:mon_fail}, where property~\ref{prop:slight:mon} fails, and properties~\ref{prop:slight:mlr} and \ref{prop:slight:slight} are satisfied (recall that for a single-task firm, every increase of informativeness is slight). Recall that given a skill distribution $p$ and perception $q$, the change in average pay is
$$W_{\left\{ \underline{a} \right\}}(p,q,\langle S', \pi' \rangle )
- W_{\left\{ \underline{a} \right\}}(p,q,\langle S, \pi \rangle )
= p(0) - q(0) .$$
Thus if $q_I(0) < p(0) < q_J(0)$, then \eqref{eq:narrow} fails, and properties~\ref{prop:slight:over} and \ref{prop:slight:under} hold.

For property~\ref{prop:slight:mlr}, return to \Cref{ex:mlr_fail}, where property~\ref{prop:slight:mlr} fails, and properties~\ref{prop:slight:mon} and \ref{prop:slight:slight} are satisfied (again recall that for a single-task firm, every increase of informativeness is slight). Given a skill distribution $p$ and perception $q$, the change in average pay is
\begin{multline*}
W_{\left\{ \overline{a} \right\}}(p,q,\langle S', \pi' \rangle )
- W_{\left\{ \overline{a} \right\}}(p,q,\langle S, \pi \rangle )
\\
\begin{aligned}
&= \left[ p(1) \cdot 1 + (1-p(1)) \cdot \left( \frac{q(0)}{q(0)+q(2)} \cdot 0 + \frac{q(2)}{q(0)+q(2)} \cdot 2 \right) \right]
\\
&\qquad - \left[ q(1) \cdot 1 + q(2) \cdot 2 \right] 
\\
&= [q(1) - p(1)] \left[ \frac{2q(2)}{1-q(1)} - 1 \right] ,
\end{aligned}
\end{multline*}
so \eqref{eq:narrow} fails if $q_I(1) > p(1) > q_J(1)$ and $\frac{q_I(2)}{1-q_I(1)} > \frac{1}{2} < \frac{q_J(2)}{1-q_J(1)}$. Let
$$\bigl( q_J(\theta), p(\theta), q_I(\theta) \bigr) \coloneqq
\begin{cases}
\left( \frac{1}{4}, \frac{1}{4}-3\delta, \frac{1}{4}-6\delta \right) & \text{for $\theta=0$} \\
\left( \frac{1}{4}, \frac{1}{4}+ \delta, \frac{1}{4}+2\delta \right) & \text{for $\theta=1$} \\
\left( \frac{1}{2}, \frac{1}{2}+2\delta, \frac{1}{2}+4\delta \right) & \text{for $\theta=2$}
\end{cases}$$
where $0 < \delta < 1/24$; then \eqref{eq:narrow} fails, and properties~\ref{prop:slight:over} and \ref{prop:slight:under} both hold.

For property~\ref{prop:slight:over}, suppose that skill types are binary, $\Theta = \{0,1\}$. Let $\overline{a} \in \mathcal{A}_M$ be the task given by $\overline{a}(\theta) \coloneqq \theta$ for each $\theta \in \Theta$, and consider the firm $A = \{\overline{a}\}$; then properties~\ref{prop:slight:mon} and \ref{prop:slight:slight} hold (again recall that for a single-task firm, every increase of informativeness is slight). Let $\langle S, \pi \rangle $ be an uninformative signal structure, i.e. one such that $\pi(s|\theta) = \pi(s|\theta')$ for every signal $s \in S$ and all skill types $\theta,\theta' \in \Theta$. Let $\langle S', \pi' \rangle $ be the signal structure given by $S' = \bigl\{s^0,s^1\bigr\}$ and $\pi'\bigl(s^0\bigm|0\bigr)=\pi'\bigl(s^1\bigm|1\bigr)=3/4$, and note that property~\ref{prop:slight:mlr} holds. Given a skill distribution $p$ and perception $q$, the change in average pay is
\begin{multline*}
W_{\left\{ \overline{a} \right\}}(p,q,\langle S', \pi' \rangle )
- W_{\left\{ \overline{a} \right\}}(p,q,\langle S, \pi \rangle )
\\
\begin{aligned}
&= \frac{ (1-p(1)) \frac{3}{4} + p(1) \frac{1}{4} }{ (1-q(1)) \frac{3}{4} + q(1) \frac{1}{4} } q(1) \tfrac{1}{4}
+ \frac{ (1-p(1)) \frac{1}{4} + p(1) \frac{3}{4} }{ (1-q(1)) \frac{1}{4} + q(1) \frac{3}{4} } q(1) \tfrac{3}{4}
- q(1) 
\\
&= \left[ \frac{ 3-2p(1) }{ 3-2q(1) }
+ 3 \frac{ 1 + 2p(1) }{ 1 + 2q(1) } - 4 \right] \frac{q(1)}{4} .
\end{aligned}
\end{multline*}
Let $p(1)=3/4$, $q_I(1)=1/4$ and $q_J(1)=1/6$; then property~\ref{prop:slight:over} fails, property~\ref{prop:slight:under} holds, and \eqref{eq:narrow} fails since
\begin{multline*}
\bigl[ W_{\left\{ \overline{a} \right\}}(p,q_I,\langle S', \pi' \rangle )
- W_{\left\{ \overline{a} \right\}}(p,q_I,\langle S, \pi \rangle ) \bigr]
\\
- \bigl[ W_{\left\{ \overline{a} \right\}}(p,q_J,\langle S', \pi' \rangle )
- W_{\left\{ \overline{a} \right\}}(p,q_J,\langle S, \pi \rangle ) \bigr]
= \frac{1}{10} - \frac{35}{384} > 0 .
\end{multline*}

For property~\ref{prop:slight:under}, apply the same argument as for \ref{prop:slight:over}, except with $p(1)=1/4$, $q_J(1)=3/4$ and $q_I(1)=5/6$.

For property~\ref{prop:slight:slight}, suppose that skill types are binary, $\Theta = \{0,1\}$, and let the skill distribution $p$ satisfy $p(0)=p(1)=1/2$. Define $\overline{a},\widetilde{a} \in \mathcal{A}$ by $\overline{a}(\theta) \coloneqq \theta$ and $\widetilde{a}(\theta) \coloneqq 4(2\theta-1)$ for each $\theta \in \Theta$, and consider the firm $A = \{\overline{a},\widetilde{a}\}$, noting that property~\ref{prop:slight:mon} holds. For each $\lambda \in [1/2,1]$, let $\langle S_\lambda, \pi_\lambda \rangle $ be the MLR signal structure given by $S_\lambda = \bigl\{s^0,s^1\bigr\}$ and $\pi_\lambda\bigl(s^0\bigm|0\bigr)=\pi_\lambda\bigl(s^1\bigm|1\bigr)=\lambda$. For any perception $q$ and any $\lambda \in [1/2,1]$, average pay at firm~$A$ is
\begin{multline*}
\begin{aligned}
W_A(p,q,\langle S_\lambda, \pi_\lambda \rangle )
&= \left[ \frac{1}{2} \lambda + \frac{1}{2} (1-\lambda) \right]
\max\left\{ r^0_{q,\lambda}, 4 \left( 2 r^0_{q,\lambda} - 1 \right) \right\}
\\
&\quad + \left[ \frac{1}{2} (1-\lambda) + \frac{1}{2} \lambda \right]
\max\left\{ r^1_{q,\lambda}, 4 \left( 2 r^1_{q,\lambda} - 1 \right) \right\}
\end{aligned}
\\
= \frac{1}{2} \max\left\{ r^0_{q,\lambda}, 4\left( 2 r^0_{q,\lambda} - 1 \right) \right\}
+ \frac{1}{2} \max\left\{ r^1_{q,\lambda}, 4\left( 2 r^1_{q,\lambda} - 1 \right) \right\}
\end{multline*}
where
$$r^0_{q,\lambda} \coloneqq \frac{ q(1) (1-\lambda) }{ (1-q(1)) \lambda + q(1) (1-\lambda) }
\quad \text{and} \quad
r^1_{q,\lambda} \coloneqq \frac{ q(1) \lambda }{ (1-q(1)) (1-\lambda) + q(1) \lambda } .$$
Thus if perceptions are $q_I(1)=3/4$ and $q_J(1)=1/4$, and signal structures are $\langle S, \pi \rangle = \langle S_{9/13}, \pi_{9/13} \rangle $ and $\langle S', \pi' \rangle = \langle S_{4/5}, \pi_{4/5} \rangle $, then properties~\ref{prop:slight:mlr}, \ref{prop:slight:over} and \ref{prop:slight:under} are satisfied, $\langle S', \pi' \rangle$ is more informative than $\langle S, \pi \rangle $, and \eqref{eq:narrow} fails by direct computation (as illustrated in \Cref{fig:tight-slight}). By \Cref{prop:slight}, since properties~\ref{prop:slight:mon}--\ref{prop:slight:under} hold and \eqref{eq:narrow} fails, it must be that property~\ref{prop:slight:slight} fails.
\qed

\section{Proof of \texorpdfstring{\Cref{prop:nearlyfull}}{Proposition \ref{prop:nearlyfull}}}
\label{app:prop:nearlyfull_pf}
If a signal structure $\langle S'', \pi'' \rangle $ is fully informative, meaning that for each signal $s'' \in S''$ there is a skill type $\theta_{s''} \in \Theta$ such that $\pi''(s''|\theta) = 0$ for every $\theta \in \Theta \setminus \{\theta_{s''}\}$, then
$$w_A(s'',q_I,\langle S'', \pi'' \rangle ) = \max_{a \in A} a(\theta) = w_A(s'',q_J,\langle S'', \pi'' \rangle ) \quad \text{for every $s'' \in S''$,}$$
so average pay is equal: $W_A(p,q_I,\langle S'',\pi'' \rangle ) - W_A(p,q_J,\langle S'',\pi'' \rangle ) = 0$.

Write $\eta \coloneqq W_A(p,q_I,\langle S,\pi \rangle ) - W_A(p,q_J,\langle S,\pi \rangle )$. Since $q_I \succsim_{LR} q_J$, we have $\eta \geq 0$ by \Cref{le:favor}. In case $\eta=0$, let $\varepsilon=0$; then any signal structure $\langle S', \pi' \rangle $ that is within $\varepsilon$ of full information is itself fully informative, so $W_A(p,q_I,\langle S',\pi' \rangle ) - W_A(p,q_J,\langle S',\pi' \rangle ) = 0 = \eta$.

Assume for the remainder that $\eta > 0$. For any $\delta \in [0,1]$, call a signal structure $\langle S', \pi' \rangle $ \emph{$\delta$-extreme} if $q_{K,\langle S', \pi' \rangle}(\theta|s') \in [0,\delta] \cup [1-\delta,1]$ for both populations $K \in \{I,J\}$, every signal $s' \in S'$, and every skill type $\theta \in \Theta$. By inspection, pay as a function of the posterior belief, $r \mapsto \max_{a \in A} \sum_{\theta \in \Theta} r(\theta) a(\theta)$, is continuous. Since any fully informative signal structure $\langle S'', \pi'' \rangle $ is $0$-extreme and satisfies $W_A(p,q_I,\langle S'',\pi'' \rangle ) - W_A(p,q_J,\langle S'',\pi'' \rangle ) = 0$, it follows that there exists a $\delta \in (0,1]$ such that any $\delta$-extreme signal structure $\langle S', \pi' \rangle $ satisfies $W_A(p,q_I,\langle S',\pi' \rangle ) - W_A(p,q_J,\langle S',\pi' \rangle ) \leq \eta$.

It remains only to show that for any $\delta \in (0,1]$, there exists an $\varepsilon>0$ such that any signal structure $\langle S',\pi'\rangle$ that is within $\varepsilon$ of full information is $\delta$-extreme. To that end, fix a $\delta \in (0,1]$, and define
$$\varepsilon \coloneqq \frac{\delta}{1-\delta} \frac{\min_{\theta \in \Theta} (\min\left\{ q_I(\theta), q_J(\theta) \right\})}{\lvert \Theta \rvert - 1} .$$
$\varepsilon$ is well-defined and strictly positive since $\lvert \Theta \rvert \geq 2$ and $q_I$ and $q_J$ have full support. Fix a signal structure $\langle S',\pi'\rangle$ that is within $\varepsilon$ of full information. Then for each $s' \in S'$, there exists a skill type $\theta_{s'} \in \Theta$ such that $\pi'(s'|\theta) \leq \varepsilon \pi'(s'|\theta_{s'})$ for every other skill type $\theta \in \Theta \setminus \{\theta_{s'}\}$, whence for any $K \in \{I,J\}$,
\begin{align*}
q_{K,\langle S',\pi'\rangle}(\theta_{s'}|s')
&= \frac{ 1 }{ 1 + \left. \left( \sum_{\theta \in \Theta \setminus \{\theta_{s'}\}} \frac{\pi'(s'|\theta)}{\pi'(s'|\theta_{s'})} q_K(\theta) \right) \middle/ q_K(\theta_{s'}) \right. }
\\
&\geq \frac{ 1 }{ 1 + (\lvert \Theta \rvert - 1) \varepsilon / [ \min_{\theta \in \Theta} q_K(\theta) ] }
\\
&\geq \frac{ 1 }{ 1 + (\lvert \Theta \rvert - 1) \varepsilon / [ \min_{\theta \in \Theta} (\min\left\{ q_I(\theta), q_J(\theta) \right\}) ] }
= 1 -\delta 
\end{align*}
and (thus) $q_{K,\langle S',\pi'\rangle}(\theta|s') \leq 1 - q_{K,\langle S',\pi'\rangle}(\theta_{s'}|s') \leq
\delta$ for every $\theta \in \Theta \setminus \{\theta_{s'}\}$.
\qed

\end{appendices}



\printbibliography[heading=bibintoc]


\end{document}

%% file: preamble.tex
\makeatletter
\renewcommand{\paragraph}{%
	\@startsection{paragraph}{4}%
	{\z@}{1.75ex \@plus 1ex \@minus .2ex}{-0.7em}%
	{\normalfont\normalsize\bfseries}%
}
\makeatother

\usepackage[T1]{fontenc}
\usepackage{lmodern}
\usepackage[utf8]{inputenc}

\usepackage[american,german,british]{babel}

\usepackage[activate={true,nocompatibility}]{microtype}

\usepackage{csquotes}

\usepackage[backend=biber,autolang=other,style=apa,maxcitenames=5,sortcites=false]{biblatex}
\DeclareLanguageMapping{british}{british-apa}
\DeclareLanguageMapping{american}{american-apa}
\usepackage{amsmath}
\usepackage{amssymb}
\usepackage{amsfonts}
\usepackage{amsthm}
\usepackage{mathtools}
\usepackage{stmaryrd}
\usepackage{centernot}

\let\originalleft\left
\let\originalright\right
\renewcommand{\left}{\mathopen{}\mathclose\bgroup\originalleft}
\renewcommand{\right}{\aftergroup\egroup\originalright}

\usepackage{graphicx}
\usepackage{tikz,pgfplots}
\pgfplotsset{compat=1.10}
\usetikzlibrary{shapes.multipart}
\usetikzlibrary{decorations.markings}
\usetikzlibrary{calc}

\usepackage[title]{appendix}

\usepackage{datetime}
\newdateformat{datestyle}{ {\THEDAY} \monthname[\THEMONTH] \THEYEAR }

\usepackage{enumerate}
\usepackage{enumitem}
\setlist[enumerate,1]{label=(\arabic*)}
\setlist[itemize,1]{label=--}
\setlist[itemize,2]{label=--}
\setlist[itemize,3]{label=--}
\setlist[itemize,4]{label=--}

\usepackage{xcolor}
\usepackage{verbatim}
\usepackage{gensymb}
\usepackage{rotating}
\usepackage{subcaption}
\usepackage{floatpag}
\floatpagestyle{plain}
\usepackage{marvosym}
\usepackage{multirow}
\usepackage{diagbox}

\usepackage{hyphenat}
\theoremstyle{definition}

\newtheorem{theorem}{Theorem}
\newtheorem{proposition}{Proposition}
\newtheorem{lemma}{Lemma}
\newtheorem{corollary}{Corollary}
\newtheorem{remark}{Remark}
\newtheorem{example}{Example}
\newtheorem{definition}{Definition}

\newtheoremstyle{named}
	{\topsep}					
	{\topsep}					
	{}							
	{0pt}						
	{\bfseries}					
	{}							
	{5pt plus 1pt minus 1pt}	
	{\thmnote{#3}}				
\theoremstyle{named}
\newtheorem{namedthm}{}

\renewcommand{\geq}{\geqslant}

\renewcommand{\leq}{\leqslant}

\usepackage{xpatch}
\xpatchcmd{\proof}{\itshape}{\proofheadfont}{}{}
\newcommand{\proofheadfont}{\slshape}

\usepackage{hyperref}
\hypersetup{pdfborder=0 0 0}

\usepackage[nameinlink]{cleveref}
\crefname{page}{p.}{pp.}
\crefname{equation}{equation}{equations}
\crefname{section}{section}{sections}
\crefname{subsection}{section}{sections}
\crefname{subsubsection}{section}{sections}
\crefname{appsec}{appendix}{appendices}
\crefname{supplsec}{supplemental appendix}{supplemental appendices}
\crefname{footnote}{footnote}{footnotes}
\crefname{figure}{figure}{figures}
\crefname{table}{table}{tables}
\crefname{theorem}{theorem}{theorems}
\crefname{proposition}{proposition}{propositions}
\crefname{lemma}{lemma}{lemmata}
\crefname{corollary}{corollary}{corollaries}
\crefname{remark}{remark}{remarks}
\crefname{observation}{observation}{observations}
\crefname{example}{example}{examples}
\crefname{fact}{fact}{facts}
\crefname{definition}{definition}{definitions}
\crefname{assumption}{assumption}{assumptions}
\crefname{exercise}{exercise}{exercises}
\crefname{notation}{notation}{notation}
\crefname{claim}{claim}{claims}
\crefname{conjecture}{conjecture}{conjectures}

\DeclareMathOperator*{\argmax}{arg\,max}

\newcommand*{\xslant}[2][76]{%
	\begingroup
	\sbox0{#2}%
	\pgfmathsetlengthmacro\wdslant{\the\wd0 + cos(#1)*\the\wd0}%
	\leavevmode
	\hbox to \wdslant{\hss
		\tikz[
			baseline=(X.base),
			inner sep=0pt,
			transform canvas={xslant=cos(#1)},
		] \node (X) {\usebox0};%
		\hss
		\vrule width 0pt height\ht0 depth\dp0 %
	}%
	\endgroup
}
\makeatletter
\newcommand*{\xslantmath}{}
\def\xslantmath#1#{%
	\@xslantmath{#1}%
}
\newcommand*{\@xslantmath}[2]{%
	\ensuremath{%
		\mathpalette{\@@xslantmath{#1}}{#2}%
	}%
}
\newcommand*{\@@xslantmath}[3]{%
	\xslant#1{$#2#3\m@th$}%
}
\makeatother

\makeatletter
\def\namedlabel#1#2{\begingroup
	#2%
	\def\@currentlabel{#2}%
	\phantomsection\label{#1}\endgroup
}
\makeatother

\makeatletter
\let\save@mathaccent\mathaccent
\newcommand*\if@single[3]{%
	\setbox0\hbox{${\mathaccent"0362{#1}}^H$}%
	\setbox2\hbox{${\mathaccent"0362{\kern0pt#1}}^H$}%
	\ifdim\ht0=\ht2 #3\else #2\fi
	}
\newcommand*\rel@kern[1]{\kern#1\dimexpr\macc@kerna}
\newcommand*\widebar[1]{\@ifnextchar^{{\wide@bar{#1}{0}}}{\wide@bar{#1}{1}}}
\newcommand*\wide@bar[2]{\if@single{#1}{\wide@bar@{#1}{#2}{1}}{\wide@bar@{#1}{#2}{2}}}
\newcommand*\wide@bar@[3]{%
	\begingroup
	\def\mathaccent##1##2{%
	  \let\mathaccent\save@mathaccent
	  \if#32 \let\macc@nucleus\first@char \fi
	  \setbox\z@\hbox{$\macc@style{\macc@nucleus}_{}$}%
	  \setbox\tw@\hbox{$\macc@style{\macc@nucleus}{}_{}$}%
	  \dimen@\wd\tw@
	  \advance\dimen@-\wd\z@
	  \divide\dimen@ 3
	  \@tempdima\wd\tw@
	  \advance\@tempdima-\scriptspace
	  \divide\@tempdima 10
	  \advance\dimen@-\@tempdima
	  \ifdim\dimen@>\z@ \dimen@0pt\fi
	  \rel@kern{0.6}\kern-\dimen@
	  \if#31
	    \overline{\rel@kern{-0.6}\kern\dimen@\macc@nucleus\rel@kern{0.4}\kern\dimen@}%
	    \advance\dimen@0.4\dimexpr\macc@kerna
	    \let\final@kern#2%
	    \ifdim\dimen@<\z@ \let\final@kern1\fi
	    \if\final@kern1 \kern-\dimen@\fi
	  \else
	    \overline{\rel@kern{-0.6}\kern\dimen@#1}%
	  \fi
	}%
	\macc@depth\@ne
	\let\math@bgroup\@empty \let\math@egroup\macc@set@skewchar
	\mathsurround\z@ \frozen@everymath{\mathgroup\macc@group\relax}%
	\macc@set@skewchar\relax
	\let\mathaccentV\macc@nested@a
	\if#31
	  \macc@nested@a\relax111{#1}%
	\else
	  \def\gobble@till@marker##1\endmarker{}%
	  \futurelet\first@char\gobble@till@marker#1\endmarker
	  \ifcat\noexpand\first@char A\else
	    \def\first@char{}%
	  \fi
	  \macc@nested@a\relax111{\first@char}%
	\fi
	\endgroup
}
\makeatother
